\pdfoutput=1
\documentclass[acmsmall,screen]{acmart}

\usepackage[utf8]{inputenc} 
\usepackage{amsmath, amsthm, amssymb}
\usepackage{thmtools}
\usepackage[english]{babel}
\usepackage{pdfpages}
\usepackage{paralist}
\usepackage{graphicx}
\usepackage{subcaption}
\usepackage{longtable}
\usepackage{booktabs}
\usepackage{tabu}
\usepackage[referable]{threeparttablex}
\usepackage{varwidth}
\usepackage{multirow}
\usepackage[ruled, linesnumbered]{algorithm2e}

\usepackage{hyperref}
\usepackage{thm-restate}
\usepackage[capitalise]{cleveref}

\crefname{figure}{Figure}{Figure}
\crefformat{footnote}{#2\footnotemark[#1]#3}
\usepackage{tikz}
\usepackage{comment}
\usepackage{parskip}

\usepackage{todonotes}
\usepackage{tikz}
\usetikzlibrary{arrows,automata,shapes,decorations,decorations.markings,calc, matrix,decorations.pathmorphing, patterns,backgrounds}

\usepackage{bm}
\usepackage{pgfplots}
\usepackage{enumerate,paralist}
\usepackage{pbox}

\usepackage{appendix}
\usepackage{calc}
\usepackage{fp}
\usepackage{multirow}
\usepackage{pbox}

\theoremstyle{theorem}
\newtheorem{remark}{Remark}

\DeclareMathAlphabet{\mathpzc}{OT1}{pzc}{m}{it}

\usepackage{graphicx} 

\mathchardef\mhyphen="2D 
\newcommand{\Nats}{\mathbb{N}}
\newcommand{\Natsplus}{\Nats^{+}}

\newcommand{\Index}{\mathsf{in}}

\newcommand{\EnumExplore}{\mathsf{DC\mhyphen DPOR}}

\newcommand{\EnumExploreCyclic}{\mathsf{DC\mhyphen DPOR \mhyphen Cyclic}}

\newcommand{\Annotation}{\mathsf{A}}
\newcommand{\Annotationpos}{\Annotation^{+}}
\newcommand{\Annotationneg}{\Annotation^{-}}

\newcommand{\Confl}{\mathsf{Confl}}
\newcommand{\ConflRW}{\mathsf{ConflRW}}
\newcommand{\Obs}{\mathsf{O}}

\newcommand{\Trace}{t}
\newcommand{\TraceInit}{\Trace^{\mathcal{I}}}
\newcommand{\SysAcquires}{\mathcal{L}^A}
\newcommand{\SysReleases}{\mathcal{L}^R}
\newcommand{\SysReads}{\mathcal{R}}
\newcommand{\SysWrites}{\mathcal{W}}
\newcommand{\Read}{r}
\newcommand{\Write}{w}
\newcommand{\Event}{e}
\newcommand{\Init}{\mathsf{Init}}
\newcommand{\HB}[3]{#1\mathsf{\to}_{#2}#3}
\newcommand{\Past}{\mathsf{Past}}
\newcommand{\Future}{\mathsf{Future}}

\newcommand{\APast}{\mathsf{Apast}}

\newcommand{\PS}{\mathsf{PS}}

\newcommand{\Value}{\mathsf{val}}

\newcommand{\Enabled}{\mathsf{enabled}}
\newcommand{\Conv}{\ast}

\newcommand{\Events}[1]{\SysEvents(#1)}
\newcommand{\Reads}[1]{\SysReads(#1)}
\newcommand{\Writes}[1]{\SysWrites(#1)}
\newcommand\pto{\mathrel{\ooalign{\hfil$\mapstochar$\hfil\cr$\to$\cr}}}
\newcommand{\Project}{|}
\newcommand{\Domain}{\mathsf{dom}}
\newcommand{\ValueDomain}{\mathcal{D}}
\newcommand{\Image}{\mathsf{img}}
\newcommand{\SeqTrace}{\tau}
\newcommand{\Realize}{\mathsf{Realize}}
\newcommand{\True}{\mathsf{True}}
\newcommand{\False}{\mathsf{False}}

\newcommand{\System}{\mathcal{P}}
\newcommand{\Process}{p}
\newcommand{\Proc}{\mathsf{proc}}

\newcommand{\Root}{r}
\newcommand{\StateSpace}{\mathcal{S}_{\System}}
\newcommand{\Locals}{\mathcal{V}}
\newcommand{\Globals}{\mathcal{G}}
\newcommand{\Locks}{\mathcal{L}}
\newcommand{\ObservedLocks}{\mathcal{L^O}}
\newcommand{\Class}[2]{[#1]_{#2}}

\newcommand{\Location}{\mathsf{loc}}
\newcommand{\CFG}{\mathsf{CFG}}
\newcommand{\SysEvents}{\mathcal{E}}
\newcommand{\SysEventsInit}{\SysWrites^{\mathcal{I}}}
\newcommand{\State}{s}
\newcommand{\Transition}{\Delta}
\newcommand{\TransSystem}{\mathcal{A}_{\System}}
\newcommand{\Path}{\rightsquigarrow}
\newcommand{\TraceSpace}{\mathcal{T}_{\System}}
\newcommand{\TraceSpaceMax}{\mathcal{T^{\max}}_{\System}}
\newcommand{\Withdraw }{\mathsf{withdraw}}

\newcommand{\NP}{NP}
\newcommand{\NPC}{NP-complete}
\newcommand{\NPH}{NP-hard}
\newcommand{\SATMONOTONE}{MONOTONE ONE-IN-THREE SAT}
\newcommand{\AEA}{ACYCLIC EDGE ADDITION}
\newcommand{\UAEA}{UNIQUE ACYCLIC EDGE ADDITION}

\newcommand{\Weight}{\lambda}

\makeatletter
\newcommand{\removelatexerror}{\let\@latex@error\@gobble}
\makeatother

\setcopyright{rightsretained} 
\acmJournal{PACMPL}
\acmPrice{}
\acmDOI{10.1145/3158119}
\acmYear{2018}
\copyrightyear{2018}
\acmJournal{PACMPL}
\acmVolume{2}
\acmNumber{POPL}
\acmArticle{31}
\acmMonth{1}

\begin{document}

\title{Data-Centric Dynamic Partial Order Reduction}

\author[M. Chalupa]{Marek Chalupa}
\affiliation{
  \institution{Masaryk University}
  \streetaddress{}
  \city{Brno}
  \postcode{}
  \country{Czech Republic}
}
\email{mchalupa@mail.muni.cz} 

\author[K. Chatterjee]{Krishnendu Chatterjee}
\affiliation{
  \institution{Institute of Science and Technology, Austria}
  \streetaddress{Am Campus 1}
  \city{Klosterneuburg}
  \postcode{3400}
  \country{Austria}
}
\email{krishnendu.chatterjee@ist.ac.at}

\author[A. Pavlogiannis]{Andreas Pavlogiannis}
\affiliation{
  \institution{Institute of Science and Technology, Austria}
  \streetaddress{Am Campus 1}
  \city{Klosterneuburg}
  \postcode{3400}
  \country{Austria}
}
\email{pavlogiannis@ist.ac.at}

\author[N. Sinha]{Nishant Sinha}
\affiliation{
  \institution{Kena Labs}
  \streetaddress{}
  \city{}
  \postcode{}
  \country{India}
}
\email{nishantsinha@acm.org}

\author[K. Vaidya]{Kapil Vaidya}
\affiliation{
  \institution{Indian Institute of Technology, Bombay}
  \streetaddress{IIT Area, Powai}
  \city{Mumbai}
  \postcode{400076}
  \country{India}
}

\begin{abstract}
We present a new dynamic partial-order reduction method for stateless model checking of concurrent programs.
A common approach for exploring program behaviors relies on enumerating 
the traces of the program, without storing the visited states (aka \emph{stateless} exploration).
As the number of distinct traces grows exponentially, dynamic 
partial-order reduction (DPOR) techniques have been successfully used to 
partition the space of traces into equivalence classes (\emph{Mazurkiewicz} 
partitioning), with the goal of exploring only few representative traces from each class.

We introduce a new equivalence on traces under sequential consistency semantics, 
which we call the \emph{observation} equivalence.
Two traces are observationally equivalent if every read event observes 
the same write event in both traces.
While the traditional Mazurkiewicz equivalence is control-centric, our new definition
is data-centric.
We show that our observation equivalence is coarser than the Mazurkiewicz equivalence, 
and in many cases even exponentially coarser.
We devise a DPOR exploration of the trace space,  called \emph{data-centric} DPOR, 
based on the observation equivalence.
\begin{compactenum}
\item For acyclic architectures, our algorithm is guaranteed to explore 
{\em exactly} one representative trace from each observation class, 
while spending polynomial time per class.
Hence, our algorithm is \emph{optimal} wrt the observation equivalence, 
and in several cases explores exponentially fewer traces than \emph{any} 
enumerative method based on the Mazurkiewicz equivalence. 
\item For cyclic architectures, we consider an equivalence between traces 
which is finer than the observation equivalence; but coarser 
than the Mazurkiewicz equivalence, and in some cases is exponentially coarser.
Our data-centric DPOR algorithm remains optimal under this trace equivalence.
\end{compactenum}
Finally, we perform a basic experimental comparison between the 
existing Mazurkiewicz-based DPOR and our data-centric DPOR
on a set of academic benchmarks.
Our results show a significant reduction in both running time and the number of explored equivalence classes.
\end{abstract}



\parskip=0.0\baselineskip \advance\parskip by 0pt plus 0pt
\maketitle
\parskip=0.5\baselineskip \advance\parskip by 0pt plus 2pt

\section{Introduction}\label{sec:intro}

\noindent{\em Stateless model-checking of concurrent programs.}
The verification of concurrent programs is one of the major 
challenges in formal methods. 
Due to the combinatorial explosion on the number of interleavings, 
errors found by testing are hard to reproduce (often called 
{\em Heisenbugs}~\cite{Musuvathi08}), and the problem needs to be addressed by a 
systematic exploration of the state space.
{\em Model checking}~\cite{Clarke00} addresses this issue, however, since model checkers
store a large number of global states, it cannot be applied to realistic
programs.
One solution that is adopted is {\em stateless model checking}~\cite{G96},
which avoids the above problem by exploring the state space without 
explicitly storing the global states.
This is typically achieved by a scheduler, which drives the program execution based on the current interaction between the processes.
Well-known tools such as VeriSoft~\cite{Godefroid97,Godefroid05} and {\sc CHESS}~\cite{Musuvathi07b}
have successfully employed stateless model checking.

\noindent{\em Partial-Order Reduction (POR).}
Even though stateless model-checking addresses the global state space issue,
it still suffers from the combinatorial explosion of the number of interleavings,
which grows exponentially.
While there are many approaches to reduce the number of explored interleavings,
such as, depth-bounding and context bounding~\cite{Lal09,Musuvathi07}, 
the most well-known method is {\em partial order reduction (POR)}~\cite{Clarke99,G96,Peled93}.
The principle of POR is that two interleavings can be regarded as equivalent 
if one can be obtained from the other by swapping adjacent, non-conflicting 
(independent) execution steps. 
The theoretical foundation of POR is the equivalence class of traces induced 
by the {\em Mazurkiewicz trace equivalence}~\cite{Mazurkiewicz87}, and 
POR explores at least one trace from each equivalence class.
POR provides a full coverage of all behaviors that can occur
in any interleaving, even though it explores only a subset of traces.
Moreover, POR is sufficient for checking most of the interesting verification
properties such as safety properties, race freedom, absence of global 
deadlocks, and absence of assertion violations~\cite{G96}.

\noindent{\em  Dynamic Partial-order Reduction (DPOR).}
Dynamic partial-order reduction (DPOR)~\cite{Flanagan05} improves the precision of POR 
by recording actually occurring conflicts during the exploration and using this
information on-the-fly. 
DPOR guarantees the exploration of at least one trace in each
Mazurkiewicz equivalence class when the explored state space is acyclic and
finite, which holds for stateless model checking, as usually the length of executions is bounded~\cite{Flanagan05,Godefroid05,Musuvathi08}.
Recently, an optimal method for DPOR was developed~\cite{Abdulla14}.
We refer to Section~\ref{sec:related} for more detailed references to related work.

\noindent{\em A fundamental limitation.} 
All existing approaches for DPOR are based on the Mazurkiewicz equivalence, i.e.,
they explore at least one (and possibly more) trace from each equivalence
class. 
A basic and fundamental question is whether coarser equivalence classes than
the Mazurkiewicz equivalence can be applied to stateless model checking 
and whether some DPOR-like approach can be developed based on such coarser equivalences. 
We start with a motivating example. 

\subsection{A Minimal Motivating Example}\label{subsec:example}

Consider a concurrent system that consists of two processes and a single global variable $x$ shown in Figure~\ref{fig:m_comparison_intro}.
\begin{figure}[!t]
\centering
\small
\begin{subfigure}[t]{0.15\textwidth}
\begin{align*}
\text{Process}&~\Process_1:\\
\hline\\[-1em]
1.&~\mathsf{write}~x;\\
2.&~\mathsf{read}~x;
\end{align*}
\end{subfigure}
\quad
\begin{subfigure}[t]{0.15\textwidth}
\begin{align*}
\text{Process}&~\Process_2:\\
\hline\\[-1em]
1.&~\mathsf{write}~x;\\
2.&~\mathsf{read}~x;
\end{align*}
\end{subfigure}
\caption{A system of two processes with two events each.}
\label{fig:m_comparison_intro}
\end{figure}

Denote by $\Write_i$ and $\Read_i$ the write and read events to $x$ by process $\Process_i$, respectively.
The system consists of four events which are all pairwise dependent, except for the pair $\Read_1,\Read_2$.
Two traces $\Trace$ and $\Trace'$ are called Mazurkiewicz equivalent, denoted $\Trace \sim_M \Trace'$, if they agree on the order of dependent events.
The traditional DPOR based on the Mazurkiewicz equivalence $\sim_M$ will explore at least one representative trace from every class induced on the trace space by the Mazurkiewicz equivalence.
There exist $\frac{2^3}{2}=4$ possible orderings of dependent events, as there are $2^3$ possible interleavings, but half of those reorder the independent events $\Read_1,\Read_2$, and thus will not be considered.
The traditional DPOR will explore the following four traces.
\begin{align*}
\Trace_1&:~\Write_1, \Read_1, \Write_2, \Read_2\hspace*{1cm}&\Trace_2:~\Write_1,\Write_2, \Read_1, \Read_2\\
\Trace_3&:~\Write_2, \Write_1, \Read_1, \Read_2&\Trace_4:~\Write_2, \Read_2, \Write_1, \Read_1
\end{align*}
Note however that $\Trace_1$ and $\Trace_4$ are state-equivalent,
in the sense that the local states visited by $\Process_1$ and $\Process_2$ are identical in the two traces.
This is because each read event \emph{observes} the same write event in $\Trace_1$ and $\Trace_4$.
In contrast, in every pair of traces among $\Trace_1,\Trace_2,\Trace_3$, there is at least one read event that observes a different write event
in that pair.
This observation makes it natural to consider two traces equivalent if they contain the same read events, and every read event observes the same write event in both traces.
This example illustrates that it is possible to have coarser equivalence than the traditional  Mazurkiewicz equivalence.

\subsection{Our Contributions}
In this work our contributions are as follows.

\noindent{\bf Observation equivalence.}
We introduce a new notion of {\em observation equivalence} 
(Section~\ref{subsec:comparison}), which is intuitively as follows:
An observation function of a trace maps every read event to the write event it
observes under sequentially consistent semantics. 
In contrast to every possible ordering of dependent control locations of 
Mazurkiewicz equivalence, in observation equivalence two traces are 
equivalent if they have the same observation function. 
The observation equivalence has the following properties.
\begin{compactenum}
\item \emph{Soundness.} 
The observation equivalence is sufficient for exploring all local states of each process,
and is thus sufficient for model checking wrt to local properties (similar to Mazurkiewicz equivalence).

\item \emph{Coarser.}
Second, we show that observation equivalence is coarser than Mazurkiewicz equivalence, i.e., 
if two traces are  Mazurkiewicz equivalent, then they are also observation equivalent 
(Section~\ref{subsec:comparison}).

\item \emph{Exponentially coarser.}
Third, we show that observation equivalence can be exponentially more succinct than
Mazurkiewicz equivalence, i.e., we present examples where the ratio of the 
number of equivalence classes between observation and  Mazurkiewicz 
equivalence is exponentially small (Section~\ref{subsec:exponential_succinct}).

\end{compactenum}
In summary, observation equivalence is a sound method which is always coarser,
and in cases, strictly coarser than the fundamental  Mazurkiewicz equivalence.

\noindent{\em Principal difference.}
The principal difference between the Mazurkiewicz and our new observation equivalence
is that while the Mazurkiewicz equivalence is {\em control-centric}, observation equivalence is {\em data-centric}.
The data-centric approach takes into account read-write and memory consistency restrictions as opposed to the
event-dependency relation of the Mazurkiewicz equivalence.

\noindent{\bf Data-centric DPOR.}
We devise a DPOR exploration of the trace space,  called \emph{data-centric} DPOR, 
based on the observation equivalence.
Our DPOR algorithm is based on a notion of {\em annotations}, which are 
intended observation functions (see Section~\ref{sec:annotations}). 
The basic computational problem is, given an annotation, decide whether there exists 
a trace which realizes the annotation.
The complexity of the problem depends on the communication graph of the system, called the architecture.
Intuitively, the nodes of the architecture represent the processes of the concurrent system,
and there is an (undirected) edge between two nodes if the respective processes access a common shared variable.
We show that the computational problem is NP-complete in general, but for the 
important special case of {\em acyclic} architectures we present a 
polynomial-time (cubic-time) algorithm based on reduction to 2-SAT 
(details in Section~\ref{sec:annotations}).
Our algorithm has the following implications.

\begin{compactenum}
\item For acyclic architectures, our algorithm is guaranteed to explore {\em exactly 
one} representative trace from each observation equivalence class, while spending \emph{polynomial time} per class.
Hence, our algorithm is \emph{optimal} wrt the observation equivalence, 
and in several cases explores exponentially fewer traces than \emph{any} 
enumerative method based on the Mazurkiewicz equivalence 
(details in Section~\ref{sec:enumerative}).

\item For cyclic architectures, we consider an equivalence between traces 
which is finer than the observation equivalence; but coarser 
than the Mazurkiewicz equivalence, and in many cases is exponentially coarser.
For this equivalence on traces, we again present an algorithm for 
DPOR that explore  {\em exactly one} representative trace from each observation class, 
while spending \emph{polynomial time} per class.
Thus again our data-centric DPOR algorithm remains optimal under this trace 
equivalence for cyclic architectures (details in Section~\ref{sec:cyclic_architectures}).
\end{compactenum}

\noindent{\bf Experimental results.}
Finally, we perform a basic experimental comparison between the 
existing Mazurkiewicz-based DPOR and our data-centric DPOR
on a set of academic benchmarks.
Our results show a significant reduction in both running time and the 
number of explored traces.

Due to lack of space, full proofs can be found in full version of this paper~\cite{istreport}.

\section{Preliminaries}\label{sec:prelim}

In this section we introduce a simple model for concurrent programs that will be used for
stating rigorously the key ideas of our data-centric DPOR.
Similar (but syntactically richer) models have been used in~\cite{Flanagan05,Abdulla14}.
In Section~\ref{subsec:discussion} we discuss our various modeling choices and possible extensions.

\noindent{\bf Informal model.}
We consider a \emph{concurrent system} of $k$ processes under sequential consistency semantics.
For the ease of presentation, we do not allow dynamic thread creation, i.e., $k$ is fixed during any execution of the system.
Each process is defined over a set of \emph{local variables} specific to the process, and a set of \emph{global variables}, which is common for all processes.
Each process is represented as an acyclic \emph{control-flow graph}, which results from unrolling the body of the process.
A process consists of statements over the local and global variables, which we call \emph{events}.
The precise kind of such events is immaterial to our model, as we are only interested in the variables involved.
In particular, in any such event we identify the local and global variables it involves, and distinguish between the variables that the event \emph{reads} from and at most one variable that the event \emph{writes} to.
Such an event is \emph{visible} if it involves global variables, and \emph{invisible} otherwise.
We consider  that processes are \emph{deterministic}, meaning that at any given time there is at most one event that each process can execute.
Given the current state of the system, a scheduler chooses one process to execute a sequence of events that is invisibly maximal,
that is, the sequence does not end while an invisible event from that process can be taken.
The processes communicate by writing to and reading from the global variables.
The system can exhibit nondeterministic behavior which is solely attributed to the scheduler, 
by choosing nondeterministically the next process to take an invisibly maximal sequence of events from any given state.
We consider \emph{locks} as the only synchronization primitive, with the available operations being acquiring a lock
and releasing a lock.
Since richer synchronization primitives are typically built using locks, this consideration is not restrictive,
and helps with keeping the exposition of the key ideas simple.

\subsection{Concurrent Computation Model}\label{subsec:model}
Here we present our model formally. Relevant notation is summarized in Table~\ref{tab:not1}.

\noindent{\bf Relations and equivalence classes.}
A binary relation $\sim$ on a set $X$ is an equivalence relation iff $\sim$ is reflexive, symmetric and transitive. 
Given an equivalence $\sim_R$ and some $x\in X$, we denote by $\Class{x}{R}$ the equivalence class of $x$ under $\sim_R$, i.e.,
\[
\Class{x}{R}=\{y\in X:~x\sim_R y\}
\]
The \emph{quotient set} $X/\sim_{R} := \{\Class{x}{R}\ |\ x\in X\}$ of $X$ under $\sim_R$ is the set of all equivalence classes of $X$ under $\sim_R$.

\noindent{\bf Notation on functions.}
We write $f:X\pto Y$ to denote that $f$ is a partial function from $X$ to $Y$.
Given a (partial) function $f$, we denote by $\Domain(f)$ and $\Image(f)$ the domain and image set of $f$, respectively.
For technical convenience, we think of a (partial) function $f$ as a set of pairs $\{(x_i,y_i)\}_i$, meaning that $f(x_i)=y_i$ for all $i$,
and use the shorthand notation $(x,y)\in f$ to indicate that $x\in\Domain(f)$ and $f(x)=y$.
Given (partial) functions $f$ and $g$, we write $f\subseteq g$ if $\Domain(f)\subseteq \Domain(g)$ and for all $x\in \Domain(f)$ we have $f(x)=g(x)$, and $f=g$ if $f\subseteq g$ and $g\subseteq f$.
Finally, we write $f\subset g$ if $f\subseteq g$ and $f\neq g$.

\noindent{\bf Model syntax.}
We consider a \emph{concurrent architecture} $\System$ that consists of a fixed number of \emph{processes} $\Process_1,\dots, \Process_k$, i.e., there is no dynamic thread creation.
Each process $\Process_i$ is defined over a set of $n_i$ \emph{local variables} $\Locals_i$, and a set of \emph{global variables} $\Globals$, which is common for all processes.
We distinguish a set of \emph{lock variables} $\Locks\subseteq \Globals$ which are used for process synchronization.
All variables are assumed to range over a finite domain $\ValueDomain$.
Every process $\Process_i$ is represented as an acyclic control-flow graph $\CFG_i$ which results from unrolling all loops in the body of $\Process_i$. Every edge of $\CFG_i$ is labeled, and called an \emph{event}.
In particular, the architecture $\System$ is associated with a set of \emph{events} $\SysEvents$, a set of \emph{read events} (or \emph{reads}) $\SysReads\subseteq \SysEvents$, a set of \emph{write events} (or \emph{writes}) $\SysWrites\subseteq \SysEvents$.
Furthermore, locks are manipulated by a set of \emph{lock-acquire} events $\SysAcquires\subseteq \SysReads$ and a set of \emph{lock-release} events $\SysReleases\subseteq \SysWrites$, which are considered read events and write events respectively.
The control-flow graph $\CFG_i$ of process $\Process_i$ consists of events of the following types
(where $\Locals_i=\{v_1,\dots,v_{n_i}\}$, $g\in\Globals$, $l\in\Locks$,
$f_i:\ValueDomain^{n_i}\to \ValueDomain$ is a function on $n_i$ arguments,
and $b:\Locals_i^{n_i}\to\{\True,\False\}$ is a boolean function on $n_i$ arguments).
\begin{compactenum}
\item\label{item:event_type1} $\Event: v\gets \mathsf{read}~g$, in which case $\Event\in \SysReads$,
\item\label{item:event_type2} $\Event: g\gets \mathsf{write}~f(v_1,\dots,v_{n_i})$, in which case $\Event\in \SysWrites$,
\item\label{item:event_type3} $\Event: \mathsf{acquire~} l$, in which case $\Event\in \SysReads$,
\item\label{item:event_type4} $\Event: \mathsf{release~} l$, in which case $\Event\in \SysWrites$,
\item\label{item:event_type5} $\Event_1: b(v_1,\dots,v_{n_i})$.
\end{compactenum}
Each $\CFG_i$ is a directed acyclic graph with a distinguished \emph{root} node $\Root_i$,
such that there is a path $\Root_i\Path x$ to every other node $x$ of $\CFG_i$.
Each node $x$ of $\CFG_i$ has either
\begin{compactenum}
\item zero outgoing edges, or
\item one outgoing edge $(x,y)$ labeled with an event of a type listed in Item~\ref{item:event_type1}-\ref{item:event_type4}, or
\item $m\geq 2$ outgoing edges $(x,y_1), \dots, (x,y_m)$ labeled with events $\Event_j:b_j(v_1,\dots,v_{n_i})$ of Item~\ref{item:event_type5},
and such that for all values of $v_1,\dots, v_{n_i}$, we have 
$b_j(v_1,\dots,v_n)\implies \neg b_l(v_1,\dots,v_{n_i})$ for all $j\neq l$.
In this case, we call $x$ a \emph{branching} node.
\end{compactenum}

\begin{figure}
\newcommand{\distone}{1cm}
\footnotesize
\centering
\begin{tikzpicture}[->,>=stealth',shorten >=1pt,auto,node distance=\distone,
                    semithick,scale=1 ]
      
\tikzstyle{every state}=[fill=white,draw=black,text=black,font=\small, thick, inner sep=0.05cm, minimum size=0.4cm]
\tikzstyle{invis}=[fill=white,draw=white,text=white,font=\small , inner sep=-0.05cm]
\def\ystep{-0.8}
\def\xbias{5.5}

\node[] at (1,0.9) {$\Locals_i=\{v_1,\dots,v_{n_i}\}$};
\node[] at (3,0.9) {$g\in\Globals$};
\node[] at (4.2,0.9) {$l\in\Locks$};

\node[state] (v1) at (0,0*\ystep) {$x$};
\node[state] (v2) at (3,0*\ystep)  {$y$};
\draw[thick] (v1) to node[above]{$\Event: v\gets \mathsf{read}~g$} (v2);
\node[] at (\xbias, 0*\ystep) {\normalsize $\Event \in \SysReads$};

\node[state] (v3) at (0,1*\ystep) {$x$};
\node[state] (v4) at (4.3,1*\ystep)  {$y$};
\draw[thick] (v3) to node[above]{$\Event: g\gets \mathsf{write}~f(v_1,\dots,v_{n_i})$} (v4);
\node[] at (\xbias, 1*\ystep) {\normalsize $\Event \in \SysWrites$};

\node[state] (v5) at (0, 2*\ystep) {$x$};
\node[state] (v6) at (2.5, 2*\ystep) {$y$};
\draw[thick] (v5) to node[above]{$\Event: \mathsf{acquire~} l$} (v6);
\node[] at (\xbias, 2*\ystep) {\normalsize $\Event \in \SysAcquires$};

\node[state] (v7) at (0, 3*\ystep) {$x$};
\node[state] (v8) at (2.5, 3*\ystep) {$y$};
\draw[thick] (v7) to node[above]{$\Event: \mathsf{release~} l$} (v8);
\node[] at (\xbias, 3*\ystep) {\normalsize $\Event \in \SysReleases$};

\node[state] (v9) at (0, 4.5*\ystep) {$x$};
\node[state] (v10) at (3.5, 4*\ystep) {$y$};
\node[state] (v11) at (3.5, 5*\ystep) {$z$};
\draw[thick] (v9) to node[above=0.2]{$\Event_1: b_1(v_1,\dots,v_{n_i})$} (v10);
\draw[thick] (v9) to node[below=0.2]{$\Event_m: b_m(v_1,\dots, v_{n_i})$} (v11);
\node[] (d) at (3.5, 4.4*\ystep) {$\vdots$};


\end{tikzpicture}
\caption{The control-flow graph $\CFG_i$ is a sequential composition of 
these five atomic graphs.}
\label{fig:cfg_universe}
\end{figure}
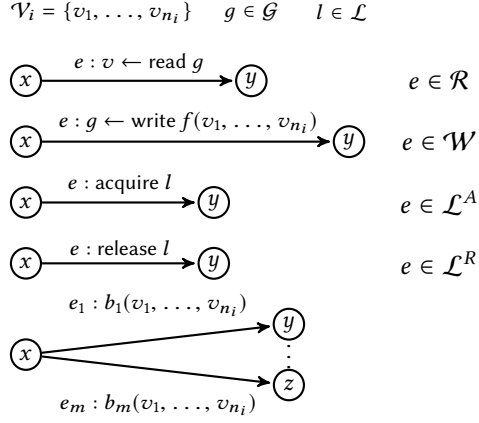
For simplicity, we require that if $x$ is a branching node, then for each edge $(x,y)$ in $\CFG_i$, the node $y$ is not branching.
Indeed, such edges can be easily contracted in a preprocessing phase.
Figure~\ref{fig:cfg_universe} provides a summary of the model syntax.
We let $\SysEvents_i\subseteq \SysEvents$ be the set of events that appear in $\CFG_i$ of process $\Process_i$,
and similarly $\SysReads_i\subseteq\SysReads$ and $\SysWrites_i\subseteq \SysWrites$ the sets of read and write events of $\Process_i$.
Additionally, we require that $\SysEvents_i\cap \SysEvents_j=\emptyset$ for all $i\neq j$ i.e., all $\SysEvents_i$ are pairwise disjoint,
and denote by $\Proc(\Event)$ the process of event $\Event$.
The \emph{location} of an event $\Location(\Event)$ is the unique global variable it involves.
Given two events $\Event, \Event'\in \SysEvents_i$ for some $\Process_i$, we write $\PS(\Event, \Event')$ if there is a path $\Event\Path \Event'$ in $\CFG_i$
(i.e., we write $\PS(\Event, \Event')$ to denote that $\Event$ is ordered before $\Event'$ in the \emph{program structure}).

We distinguish a set of \emph{initialization events} $\SysEventsInit\subseteq \SysWrites$ with $|\SysEventsInit|=|\Globals|$ which are attributed to process $\Process_1$, and are used to initialize all the global variables to some fixed values.
For every initialization write event $\Write^{\mathcal{I}}$
and for any event $\Event\in \SysEvents_i$ of process $\Process_i$,
we define that $\PS(\Write^{\mathcal{I}}, \Event)$ (i.e., the initialization events occur before any event of each process).
Figure~\ref{fig:model_example} illustrates the above definitions on the typical bank account example.
\begin{figure*}[h]
\removelatexerror
\centering
\footnotesize
\newcommand{\distone}{1cm}
\centering
\begin{tikzpicture}[->,>=stealth',shorten >=1pt,auto,node distance=\distone,
                    semithick,scale=1 ]
      
\tikzstyle{every state}=[fill=white,draw=black,text=black,font=\small, thick, inner sep=0.05cm, minimum size=0.4cm]
\tikzstyle{invis}=[fill=white,draw=white,text=white,font=\small , inner sep=-0.05cm]

\node[text width=5.5cm] at (-2,-0.35){%
\begin{algorithm}[H]
\SetKwInOut{Globals}{Globals}
\SetKwInOut{Locals}{Locals}
\setcounter{algocf}{-1}
\renewcommand{\thealgocf}{}
\SetAlgorithmName{Method}{method}
\SetAlgoNoLine
\DontPrintSemicolon
\caption{$\mathit{bool}~\Withdraw(\mathit{int}~\mathsf{amount})$}\label{algo:withdraw}
\Globals{$\mathit{int}~\mathsf{balance},~\mathit{lock}~l$}
\Locals{$\mathit{bool}~\mathsf{success},~\mathit{int}~v$}
\tcp{1. Try withdraw}
$\mathsf{success}\gets \False$\\
$\mathsf{acquire}(l)$\\
$v\gets \mathsf{balance}$\\
\uIf{$v-\mathsf{amount}\geq 0$}{
$\mathsf{balance}\gets v - \mathsf{amount}$\\
$\mathsf{success}\gets \True$
}
$\mathsf{release}(l)$\\
$\mathsf{print}(\mathsf{success})$\\
\tcp{2. Print balance}
$v\gets \mathsf{balance}$\\
$\mathsf{print}(v)$
\end{algorithm}
};

\def\xdisp{2}
\def\ystep{0.9}
\node[state] (v1) at (0+\xdisp,2-0*\ystep) {$x_1$};
\node[state] (v2) at (0+\xdisp,2-1*\ystep) {$x_2$};
\node[state] (v3) at (0+\xdisp,2-2*\ystep) {$x_3$};
\node[state] (v4) at (3.5+\xdisp, 2-2.5*\ystep) {$x_4$};
\node[state] (v5) at (0+\xdisp, 2-3*\ystep) {$x_5$};
\node[state] (v6) at (0+\xdisp, 2-4*\ystep) {$x_6$};
\node[state] (v7) at (0+\xdisp, 2-5*\ystep) {$x_7$};
\draw[thick] (v1) to node[right=0.2, pos=0.4]{$\Event_1: \mathsf{acquire}~l$} (v2);
\draw[thick] (v2) to node[right=0.2, pos=0.4]{$\Event_2: v\gets \mathsf{read}~\mathsf{balance}$} (v3);
\draw[thick, bend left=10] (v3) to node[above=0.2,pos=0.9]{$\Event_3: b(v, \mathsf{amount})$} (v4);
\draw[thick, bend left=10] (v4) to node[right=0.2,pos=0.5]{$\Event_4: \mathsf{balance}\gets \mathsf{write}~f(v,\mathsf{amount},\mathsf{success})$} (v5);
\draw[thick] (v5) to node[right=0.2, pos=0.4]{$\Event_5: \mathsf{release}~l$} (v6);
\draw[thick] (v3) to node[right=0.2]{$\Event_7: \neg b(v, \mathsf{amount})$} (v5);
\draw[thick] (v6) to node[right=0.2, pos=0.4]{$\Event_6: v\gets \mathsf{read}~\mathsf{balance}$} (v7);

\node[text width=2.8cm] at (6+\xdisp, 1.6) {\normalsize $\SysEvents_i=\{\Event_1,\dots,\Event_7\}$\\ $\SysReads_i=\{\Event_1, \Event_2, \Event_6\}$\\ $\SysWrites_i=\{\Event_4, \Event_5\}$\\ $\SysAcquires_i=\{e_1\}$\\ $\SysReleases_i=\{e_5\}$};

\end{tikzpicture}
\caption{\textit{(Left)}:~A method $\Withdraw$ executed whenever some $\mathsf{amount}$ is to be extracted from the $\mathsf{balance}$ of a bank account. \\
\textit{(Right)}:~Representation of $\Withdraw$ in our concurrent model.
The root node is $x_1$. The program structure orders $\PS(\Event_2, \Event_4)$.
We have $\Location(\Event_1)=\Location(\Event_5)$ and $\Location(\Event_2)=\Location(\Event_4)=\Location(\Event_6)$.
}
\label{fig:model_example}
\end{figure*}

\noindent{\bf Model semantics.}
A \emph{local state} of a process $\Process_i$ is a pair $\State_i=(x_i,\Value_i)$ where $x_i$ is a node of $\CFG_i$ (i.e., the program counter) and $\Value_i$ is a valuation on the local variables $\Locals_i$.
A \emph{global state} of $\System$ is a tuple $\State=(\Value, \State_1,\dots,\State_k)$, where $\Value$ is a valuation on the global variables $\Globals$ and $\State_i$ is a local state of process $\Process_i$.
An event $\Event$ along an edge $(x,y)$ of a process $\Process_i$ is \emph{enabled} in $\State$ if $\State_i=(x,\Value_i)$ (i.e., the program counter is on node $x$) and additionally,
\begin{compactenum}
\item if $\Event: \mathsf{acquire~} l$, then $\Value(l)=\False$, and
\item if $\Event: b_j(v_1,\dots,v_{n_i})$, then $b_j(\Value_i(v_1),\dots,\Value_i(v_{n_i}))=\True$.
\end{compactenum}
In words, if $\Event$ acquires a lock $l$, then $\Event$ is enabled iff $l$ is free in $\State$,
and if $x$ is a branching node, then $\Event$ is enabled iff it respects the condition of the branch in $\State$.
Given a state $\State$, we denote by $\Enabled(\State)\subseteq \SysEvents$ the set of enabled events in $\State$,
and observe that there is at most one enabled event in each state $\State$ from each process.
The execution of an enabled event $\Event$ along an edge $(x,y)$ of $\Process_i$ in state $\State=(\Value, \State_1,\dots,\State_k)$ results in a state $\State'=(\Value', \State_1,\dots, \State'_i, \dots, \State_k)$, where $\State'_i=(y, \Value'_i)$.
That is, the program counter of $\Process_i$ has progressed to $y$, and the valuation functions $\Value'$ and $\Value'_i$ have been modified according to standard semantics, as follows:
\begin{compactenum}
\item\label{item:event_sem_type1} $\Event: v\gets \mathsf{read}~g$ then $\Value'_i(v)=\Value(g)$,
\item\label{item:event_sem_type2} $\Event: g\gets \mathsf{write}~f(v_1,\dots,v_{n_i})$ then $\Value'(g)= f(\Value_i(v_1),\dots,\Value_i(v_{n_i}))$,
\item\label{item:event_sem_type3} $\Event: \mathsf{acquire~} l$ then $\Value'(l)=\True$,
\item\label{item:event_sem_type4} $\Event: \mathsf{release~} l$ then $\Value'(l)=\False$.
\end{compactenum}
Moreover, $\Value$ agrees with $\Value'$ and $\Value_i$ agrees with $\Value'_i$ on all other variables.
We write $\State\xrightarrow{\Event}\State'$ to denote that the execution of event $\Event$ in $\State$ results in state $\State'$.
Let $\StateSpace$ be the finite set (since variables range over a finite domain) of states of $\System$.
The semantics of $\System$ are defined in terms of a transition system $\TransSystem=(\StateSpace, \Transition, \State^0)$,
where $\State^0$ is the initial state, and $\Transition\subseteq  \StateSpace \times \StateSpace$ is the transition relation such that 
\[
(\State, \State') \in \TransSystem \text{ iff } \exists \Event\in\Enabled(\State):~\State\xrightarrow{\Event}\State'
\]
and either $\Event$ is an initialization event, or the program counter of $\Process_1$ has passed all initialization edges of $\Process_1$.
We write
$
\State \xrightarrow{\Event_1,\dots \Event_n} \State'
$
if there exists a sequence of states $\{\State^i\}_{1\leq i <n}$ such that
\[
\State\xrightarrow{\Event_1}\State^1\xrightarrow{\Event_2}\ldots \State^{n-1} \xrightarrow{\Event_n}\State'
\]
The initial state $\State^0=(\Value, \State^0_1,\dots, \State^0_k)$ is such that the value $\Value(g)$ of each global variable $g$ comes from the unique initialization write event $\Write$ with $\Location(\Write)=g$, and for each $\State^0_i=(x_i,\Value_i)$ we have that   $x_i=\Root_i$ (i.e., the program counter of process $\Process_i$ points to the root node of $\CFG_i$).
For simplicity we restrict $\StateSpace$ to states $\State$ that are reachable from the initial state $\State^0$ by a sequence of events
$
\State^0\xrightarrow{\Event_1,\dots,\Event_n} \State
$.
We focus our attention on state spaces $\StateSpace$ that are acyclic.

\noindent{\bf Architecture topologies.}
The architecture $\System$ induces a labeled undirected communication graph $G_{\System}=(V_{\System}, E_{\System}, \Weight_{\System})$ where $V_{\System}=\{\Process_i\}_i$.
There is an edge $(\Process_i, \Process_j)$ if processes $\Process_i,\Process_j$ access a common global variable or a common lock.
The label $\Weight(\Process_i, \Process_j)$ is the set of all such global variables and locks.
We call $\System$ \emph{acyclic} if $G_{\System}$ does not contain cycles.
The class of acyclic architectures includes, among others, all architectures with two processes, star architectures, pipelines, tree-like and hierarchical architectures.

\begin{table}[h]
\renewcommand{\arraystretch}{1.3}
\setlength\tabcolsep{2pt}
\small
\caption{Notation on the concurrent architecture.}\label{tab:not1}
\begin{center}
\begin{tabular}{c|c}
\hline
{\bf Notation} & {\bf Interpretation}\\
\hline
\hline
$\System=(\Process_i)_{i=1}^k$ & the concurrent architecture of $k$ processes\\
\hline
$\Globals,\Locals, \Locks$ & the global, local and lock variables\\
\hline
\pbox{3cm}{$\SysEvents, \SysWrites, \SysReads, \SysAcquires,$\\ $\SysReleases, \SysEventsInit$} &\pbox{10cm}{the set of events, write, read, lock-acquire\\ lock-release and initialization events} \\
\hline
$\Value_i,\Value$ & valuations of local, global variables \\
\hline
$\Enabled(\State)\subseteq \SysEvents$ & the set of enabled events in $\State$ \\
\hline
$\State \xrightarrow{\Event_1,\dots, \Event_n} \State'$ & sequence of events from $\State$ to $\State'$ \\
\hline
$\Proc(\Event)$, $\Location(\Event)$ & the process, the global variable of event $\Event$\\
\hline
$\CFG_i$, $\PS\subseteq \SysEvents\times \SysEvents$ & \pbox{10cm}{the control-flow graph of process $\Process_i$,\\ and the program structure relation}\\
\hline
$G_{\System}=(V_{\System}, E_{\System}, \Weight_{\System})$ & the communication graph of $\System$\\
\hline
\end{tabular}
\end{center}
\end{table}

\subsection{Traces}\label{subsec:traces}
In this section we develop various helpful definitions on traces. 
Relevant notation is summarized in \cref{tab:not2}.

\smallskip\noindent{\bf Notation on traces.}
A (concrete, concurrent) \emph{trace} is a sequence of events $\Trace=\Event_1,\dots,\Event_j$
such that for all $1\leq i<j$, we have $\State^{i-1}\xrightarrow{\Event_i}\State^i$, where $\State^i\in\StateSpace$ and $\State^0$ is the initial state of $\System$.
In such a case, we write succinctly $\State^0\xrightarrow{\Trace} \State^j$.
We fix the first $|\Globals|$ events $\Event_1,\dots,\Event_{|\Globals|}$ of each trace $\Trace$ to be initialization events that write the initial values to the global variables.
That is, for all $1\leq i\leq |\Globals|$ we have $\Event_i\in\SysWrites$, and hence every trace $\Trace$ starts with an \emph{initialization trace} $\TraceInit$ as a prefix.
Given a trace $\Trace$, we denote by $\Events{\Trace}$ the set of events that appear in $\Trace$,
with $\Reads{\Trace}=\Events{\Trace}\cap \SysReads$ the read events in $\Trace$, 
and with $\Writes{\Trace}=\Events{\Trace}\cap \SysWrites$ the write events in $\Trace$,
and let $|\Trace|=|\Events{\Trace}|$ be the \emph{length} of $\Trace$.
For an event $\Event\in\Events{\Trace}$, we write $\Index_{\Trace}(\Event)\in \Natsplus$ to denote the index of $\Event$ in $\Trace$.
Given some $\ell\in \Nats$, we denote by $\Trace[\ell]$ the prefix of $\Trace$ up to position $\ell$, and we say that $\Trace$ is an \emph{extension} of $\Trace[\ell]$.
We let $\Enabled(\Trace)$ denote the set of enabled events in the state at the end of $\Trace$, and call $\Trace$ \emph{maximal} if $\Enabled(\Trace)=\emptyset$.
We write $\TraceSpace$ (resp., $\TraceSpaceMax$) for the set of all traces (resp., maximal traces) of $\System$.
We denote by $\State(\Trace)$ the unique state of $\System$ such that $\State^0\xrightarrow{\Trace}\State(\Trace)$,
and given an event $\Event\in \Reads{\Trace}\cup\Writes{\Trace}$, denote by $\Value_{\Trace}(\Event)\in \ValueDomain$ the \emph{value} that the unique global variable of $\Event$ has in $\State(\Trace[\Index_{\Trace}(\Event)])$.
We call a maximal trace $\Trace$ \emph{lock-free} if the value of every lock variable in $\State(\Trace)$ is $\False$
(i.e., all locks have been released at the end of $\Trace$).
An event $\Event$ is \emph{inevitable} in a trace $\Trace$ if every every lock-free maximal extension of $\Trace$ contains $\Event$.
Given a set of events $A$, we denote by $\Trace \Project A$ the \emph{projection} of $\Trace$ on $A$, which is the unique subsequence of $\Trace$ that contains all events of $A\cap \Events{\Trace}$, and only those.
A sequence of events $\Trace'$ is called the \emph{global projection} of another sequence $\Trace$ if $\Trace'=\Trace\Project (\SysReads\cup \SysWrites)$.

\smallskip\noindent{\bf Sequential traces.}
Given a process $\Process_i$, a \emph{sequential trace} $\SeqTrace_i$  is a sequence of events that correspond to a path in $\CFG_i$, starting from the root node $\Root_i$.
Note that a sequential trace is only wrt $\CFG_i$, and is not necessarily a trace of the system.
The notation on traces is extended naturally to sequential traces (e.g., $\Events{\SeqTrace_i}$ and $\Reads{\SeqTrace_i}$ denote the events and read events of the sequential trace $\SeqTrace_i$, respectively).
Given $k$ sequential traces $\SeqTrace_1, \SeqTrace_2,\dots, \SeqTrace_k$, so that each $\SeqTrace_i$ is wrt $\Process_i$, we denote by
$
\SeqTrace_1\Conv \SeqTrace_2 \Conv \dots \Conv \SeqTrace_k
$
the (possibly empty) set of all traces $\Trace$ such that $\Events{\Trace}=\bigcup_{1\leq i\leq k} \Events{\SeqTrace_i}$.

\smallskip\noindent{\bf Conflicting events, dependent events and happens-before relations.}
Two events $\Event_1,\Event_2\in \SysReads\cup\SysWrites$ are said to \emph{conflict}, written $\Confl(\Event_1,\Event_2)$ if $\Location(\Event_1)=\Location(\Event_2)$ and at least one is a write event.
The events are said to be in \emph{read-write conflict} if $\Event_1\in \SysReads$, $\Event_2\in\SysWrites$ and $\Confl(\Event_1,\Event_2)$.
Two events $\Event_1, \Event_2$ are said to be \emph{independent}~\cite{G96,Flanagan05} if
$\Process(\Event_1)\neq \Process(\Event_2)$ and
\begin{compactenum}
\item for each $i\in \{1,2\}$ and pair of states $\State_1,\State_2$ such that $\State_1\xrightarrow{\Event_i}\State_2$, 
we have that $\Event_{3-i}\in\Enabled(\State_1)$ iff $\Event_{3-i}\in\Enabled(\State_2)$, and
\item for any pair of states $\State_1,\State_2$ such that $\Event_1,\Event_2\in \Enabled(\State_1)$,
we have that $\State_1\xrightarrow{\Event_1,\Event_2}\State_2$ iff $\State_1\xrightarrow{\Event_2,\Event_1} \State_2$,
\end{compactenum}
and \emph{dependent} otherwise.
Following the standard approach in the literature, we will consider two conflicting events to be always dependent~\cite[Chapter~3]{Godefroid97}
(e.g., two conflicting write events are dependent, even if they write the same value).
A sequence of events $\Trace$ induces a \emph{happens-before} relation $\HB{}{\Trace}{}\subseteq \Events{\Trace}\times \Events{\Trace}$,
which is the smallest transitive relation on $\Events{\Trace}$ such that
\[
\HB{\Event_1}{\Trace}{\Event_2}\quad \text{ if }\quad  \Index_{\Trace}(\Event_1)\leq \Index_{\Trace}(\Event_2) \text{ and } \Event_1 \text{ and } \Event_2 \text{ are dependent.}
\]
Observe that $\HB{}{\Trace}{}$ orders all pairwise conflicting events, as well as all the events of any process.

\begin{table}[h]
\renewcommand{\arraystretch}{1.3}
\setlength\tabcolsep{5pt}
\small
\caption{Notation on traces.}\label{tab:not2}
\begin{center}
\begin{tabular}{c|c}
\hline
{\bf Notation} & {\bf Interpretation}\\
\hline
\hline
$\Trace$, $\SeqTrace_i$ & a trace and a sequential trace\\
\hline
$\Confl(\Event_1,\Event_2)$ & conflicting events \\
\hline
$\Trace[\ell]$, $|\Trace|$ & the prefix up to index $\ell$, and length of $\Trace$\\
\hline
$\Events{\Trace}, \Writes{\Trace}, \Reads{\Trace}$  & the events, write and read events of trace $\Trace$\\
\hline
$\Index_{\Trace}(\Event)$, $\Value_{\Trace}(\Event)$ & the index and value of event $\Event$ in trace $\Trace$\\
\hline
$t\Project X$ & projection of trace $\Trace$ on event set $X$\\
\hline
$\Enabled(\Trace)$ & the enabled events in the state reached by $\Trace$\\
\hline
$\HB{}{\Trace}{}$ & the happens-before relation on $\Trace$\\
\hline
$\Obs_{\Trace}$ & the observation function of $\Trace$\\
\hline
\end{tabular}
\end{center}
\end{table}

\subsection{Discussion and Remarks}\label{subsec:discussion}

The concurrent model we consider here is minimalistic, to allow for a clear exposition of the ideas used in our data-centric DPOR.
Here we discuss some of the simplifications we have adopted to keep the presentation simple.

\noindent{\bf Global variables and arrays.}
First, note that the location $\Location(\Event)$ of every event $\Event\in \SysReads\cup \SysWrites$ is taken to be fixed in each $\CFG_i$.
The dynamic access of a static, global data structure $g$ based on the value of a local variable $v$ (e.g., accessing the element $g[v]$ of a global array $g$) can be modeled by using a different global variable $g_i$ to encode the $i$-th location of $g$,
and a sequence of branching nodes that determine which $g_i$ should be accessed based on the value of $v$.
Our framework can be strengthened to allow use of global arrays directly, and our algorithms apply straightforwardly to this richer framework. However, this would complicate the presentation, and is thus omitted in the theoretical 
exposition of the paper.
Arrays are handled naturally in our implementation,
and we refer to the Experiments Section~\ref{subsec:implementation}
for a description.

\noindent{\bf Invisible computations.}
Each process $\Process_i$ is deterministic, and the only source of nondeterminism in the executions of the system comes from a nondeterministic scheduler that chooses an enabled event to be executed from a given state.
The model uses the functions $f$ and $b$ on events $\Event: g\gets \mathsf{write}~f(v_1,\dots,v_j)$ and $\Event: b(v_1,\dots, b_n)$ respectively to collapse deterministic invisible computations of each process, and only consider the value that $f$ writes on a global variable (in addition to the side-effects that $f$ has on local the variables of process $\Process_i$).
This is a standard approach in modeling concurrent systems, as interleaving invisible events does not 
change the set of reachable local states of the processes.

\noindent{\bf Locks and synchronization mechanisms.}
We treat lock-acquire events as reads and lock-release events as writes.
In a trace $\Trace$, a lock-acquire event $\Event$ is considered to read the value of the last lock-release event $\Event'$ on the same lock $l$ (or some initialization event $\Init$ if $\Event$ is the first lock event on $l$ in $\Trace$).
Our approach can be extended to richer communication (e.g., message passing) and synchronization primitives (e.g. semaphores, wait-notify), which are often implemented using some low-level locking mechanism.

\noindent{\bf Maximal lock-free traces.}
We also assume that in every maximal trace of the system, every lock-acquire is followed by a corresponding lock-release.
Traces without this property are typically considered erroneous, and some modern programming languages even 
force this restriction syntactically.

\section{Observation Trace Equivalence}\label{sec:mazurkiewicz}

In this section, 
we introduce the observation equivalence $\sim_{\Obs}$ on traces,
upon which in the later sections we develop our data-centric DPOR.
We explore the relationship between the control-centric Mazurkiewicz equivalence $\sim_{M}$ and 
the observation equivalence.
In particular, we show that $\sim_{\Obs}$ refines $\sim_{M}$, that is, every two traces that are equivalent under reordering of independent events
are also equivalent under observations.
We conclude by showing that $\sim_{\Obs}$ can be exponentially more succinct,
both in the number of processes, and the size of each process.


\subsection{Mazurkiewicz and Observation Equivalence}\label{subsec:comparison}
In this section we introduce our notion of observation equivalence.
We start with the classical definition of Mazurkiewicz equivalence and
then the notion of observation functions.

\noindent{\bf Mazurkiewicz trace equivalence.}
Two traces $\Trace_1,\Trace_2\in\TraceSpace$ are called \emph{Mazurkiewicz equivalent} if one can be obtained from the other by swapping adjacent, independent events.
Formally, we write $\sim_{M}$ for the Mazurkiewicz equivalence on $\TraceSpace$, and we have $\Trace_1\sim_{M} \Trace_2$ iff
\begin{compactenum}
\item $\Events{\Trace_1}=\Events{\Trace_2}$, and
\item for every pair of events $\Event_1,\Event_2\in \Events{\Trace_1}$ we have that
$\HB{\Event_1}{\Trace_1}{\Event_2}$ iff $\HB{\Event_1}{\Trace_2}{\Event_2}$.
\end{compactenum}

\noindent{\bf Observation functions.}
The concurrent model introduced in Section~\ref{subsec:model} follows \emph{sequential consistency}~\cite{Lamport79}, 
i.e., all processes observe the same order of events,
and a read event of some variable will observe the value written by the last write event to that variable in this order.
Throughout the paper, an \emph{observation function} is going to be a partial function $\Obs:\SysReads \pto\SysWrites$.
A trace $\Trace$ induces a total observation function $\Obs_{\Trace}:\Reads{\Trace}\to \Writes{\Trace}$ following the sequential consistency axioms. That is, $\Obs_{\Trace}(\Read)=\Write$ iff
\begin{compactenum}
\item $\Index_{\Trace}(\Write)<\Index_{\Trace}(\Read)$, and
\item for all $\Write'\in \Writes{\Trace}$ such that $\Confl(\Read, \Write')$ we have that 
$\Index_{\Trace}(\Write')<\Index_{\Trace}(\Write) \text{ or } \Index_{\Trace}(\Write')>\Index_{\Trace}(\Read)$.
\end{compactenum}
We say that $\Trace$ is \emph{compatible} with an observation function $\Obs$ if $\Obs\subseteq \Obs_{\Trace}$,
and that $\Trace$ \emph{realizes} $\Obs$ if $\Obs = \Obs_{\Trace}$.

\noindent{\bf Observation equivalence.}
We define the \emph{observation equivalence} $\sim_{\Obs}$ on the trace space $\TraceSpace$ as follows.
For $\Trace_1,\Trace_2\in \TraceSpace$  we have $\Trace_1\sim_{\Obs}\Trace_2$ iff $\Events{\Trace_1}=\Events{\Trace_2}$ and $\Obs_{\Trace_1}=\Obs_{\Trace_2}$, 
i.e., the two observation functions coincide.

We start with the following crucial lemma.
In words, it states that if two traces agree on their observation functions, 
then they also agree on the values seen by their common read events.

\begin{restatable}{lemma}{obstoval}\label{lem:obs_to_val}
Consider two traces $\Trace_1, \Trace_2$ such that $\Obs_{\Trace_1}\subseteq \Obs_{\Trace_2}$.
Then 
\begin{itemize}
\item for all read events $\Read\in \Reads{\Trace_1}$ we have that $\Value_{\Trace_1}(\Read)=\Value_{\Trace_2}(\Read)$, and
\item for all write events $\Write\in \Writes{\Trace_1}\cap \Writes{\Trace_2}$ we have that $\Value_{\Trace_1}(\Write)=\Value_{\Trace_2}(\Write)$.
\end{itemize}
\end{restatable}

The following is an easy consequence of Lemma~\ref{lem:obs_to_val}.

\begin{lemma}\label{lem:obs_to_valmax}
Consider two traces $\Trace_1, \Trace_2$ such that $\Obs_{\Trace_1}\subseteq \Obs_{\Trace_2}$ and $\Trace_2$ is maximal.
Then (i)~$\Events{\Trace_1}\subseteq \Events{\Trace_2}$, and
(ii)~for all events $\Event\in \Reads{\Trace_1}\cup \Writes{\Trace_1}$ we have that $\Value_{\Trace_1}(\Event)=\Value_{\Trace_2}(\Event)$.
\end{lemma}

\noindent{\bf Soundness.}
Lemma~\ref{lem:obs_to_valmax} implies that two maximal traces which agree on their observation function have the same observable behavior, i.e., each global event has the same value in the two traces.
Since all local states of each process can be explored by exploring maximal traces,
it suffices to explore all the (maximal) observation functions  of $\System$.

The Mazurkiewicz trace equivalence is \emph{control-centric}, i.e., equivalent traces share the same order between the dependent control locations of the program.
In contrast, the observation trace equivalence is \emph{data-centric}, as it is based on which write events are observed by the read events of each trace.
Note that two conflicting events are dependent, and thus must be ordered in the same way by two Mazurkiewicz-equivalent traces. 
The formal relationship between the two equivalences is established in the following theorem.

\begin{theorem}\label{them:comparison}
For any two traces $\Trace_1,\Trace_2\in \TraceSpace$, if $\Trace_1\sim_{M} \Trace_2$ then $\Trace_1\sim_{\Obs} \Trace_2$.
\end{theorem}
\begin{proof}
Consider any read event $\Read\in \Reads{\Trace_1}$ and assume towards contradiction that $\Obs_{\Trace_1}(\Read)\neq \Obs_{\Trace_2}(\Read)$. 
Let $\Write_1=\Obs_{\Trace_1}(\Read)$ and $\Write_2=\Obs_{\Trace_2}(\Read)$.
Since $\Trace_1\sim_{M} \Trace_2$, we have that $\Write_1\in\Events{\Trace_2}$ and $\Write_2\in \Events{\Trace_1}$.
Then $\HB{\Write_1}{\Trace_1}{\Read}$ and $\HB{\Write_2}{\Trace_2}{\Read}$, and one of the following holds.
\begin{compactenum}
\item $\HB{\Read}{\Trace_1}{\Write_2}$, and since $\HB{\Write_2}{\Trace_2}{\Read}$ then $\Trace_1\neq \sim_{M} \Trace_2$, a contradiction.
\item $\HB{\Write_2}{\Trace_1}{\Write_1}$, and since $\Trace_1\sim_{M}\Trace_2$ we have that $\HB{\Write_2}{\Trace_2}{\Write_1}$, and thus $\HB{\Read}{\Trace_2}{\Write_1}$. Since $\HB{\Write_1}{\Trace_1}{\Read}$, we have $\Trace_1\neq \sim_{M} \Trace_2$, a contradiction.
\end{compactenum}
The desired result follows.
\end{proof}

\begin{example}[Mazurkiewizc-based-based vs observation exploration.]
\begin{figure*}[h]
\removelatexerror
\footnotesize
\centering
\newcommand{\distone}{1cm}
\centering
\begin{tikzpicture}[->,>=stealth',shorten >=1pt,auto,node distance=\distone,
                    semithick,scale=1 ]
      
\tikzstyle{every state}=[fill=white,draw=black,text=black,font=\small, thick, rectangle, inner sep=0.05cm, minimum size=0.4cm]
\tikzstyle{invis}=[fill=white,draw=white,text=white,font=\small , inner sep=-0.05cm]
\tikzstyle{dashed}=                  [dash pattern=on 6pt off 2pt]

\def\xdisp{-3}
\def\ystep{0.9}
\node[state, circle] (v1) at (0+\xdisp,3.6-0*\ystep) {$x_1$};
\node[state, circle] (v2) at (0+\xdisp,3.6-1*\ystep) {$x_2$};
\node[state, circle] (v3) at (0+\xdisp,3.6-2*\ystep) {$x_3$};
\node[state, circle] (v4) at (0+\xdisp,3.6-3*\ystep) {$x_4$};
\node[state, circle] (v5) at (0+\xdisp,3.6-5*\ystep) {$x_6$};
\node[state, circle] (v6) at (0+\xdisp,3.6-4*\ystep) {$x_5$};
\node[state, circle] (v7) at (0+\xdisp,3.6-6*\ystep) {$x_7$};
\draw[thick] (v1) to node[right=0.2, pos=0.4]{$\Event_1: \mathsf{acquire}~l$} (v2);
\draw[thick] (v2) to node[right=0.2, pos=0.4]{$\Event_2: v\gets \mathsf{read}~\mathsf{balance}$} (v3);
\draw[thick] (v3) to node[right=0.2,pos=0.4]{$\Event_3: b(v, \mathsf{amount})$} (v4);
\draw[thick] (v4) to node[right=0.2,pos=0.4]{$\Event_4: \mathsf{balance}\gets \mathsf{write}~f(v,\mathsf{amount},\mathsf{success})$} (v6);
\draw[thick] (v6) to node[right=0.2, pos=0.4]{$\Event_5: \mathsf{release}~l$} (v5);
\draw[thick] (v5) to node[right=0.2, pos=0.4]{$\Event_6: v\gets \mathsf{read}~\mathsf{balance}$} (v7);

\def\xdisp{3}
\def\ystep{0.8}
\def\bend{20}

\node[] at (0+\xdisp,2+0.75*\ystep) {\large $\Process_1$};
\node[state] (v1) at (0+\xdisp,2-0*\ystep) {$\Event_1$};
\node[state] (v2) at (0+\xdisp,2-1*\ystep) {$\Event_2$};
\node[state] (v3) at (0+\xdisp,2-2*\ystep) {$\Event_3$};
\node[state] (v4) at (0+\xdisp, 2-3*\ystep) {$\Event_4$};
\node[state] (v5) at (0+\xdisp,2-4*\ystep) {$\Event_5$};
\node[state] (v6) at (0+\xdisp,2-5*\ystep) {$\Event_6$};
\draw[thick, ] (v1) to (v2);
\draw[thick, ] (v2) to (v3);
\draw[thick, ] (v3) to (v4);
\draw[thick, ] (v4) to (v5);
\draw[thick, ] (v5) to (v6);

\def\xdisp{6}
\node[] at (0+\xdisp,2+0.75*\ystep) {\normalsize $\Process_2$};
\node[state] (u1) at (0+\xdisp,2-0*\ystep) {$\Event_1$};
\node[state] (u2) at (0+\xdisp,2-1*\ystep) {$\Event_2$};
\node[state] (u3) at (0+\xdisp,2-2*\ystep) {$\Event_3$};
\node[state] (u4) at (0+\xdisp, 2-3*\ystep) {$\Event_4$};
\node[state] (u5) at (0+\xdisp,2-4*\ystep) {$\Event_5$};
\node[state] (u6) at (0+\xdisp,2-5*\ystep) {$\Event_6$};
\draw[thick, ] (u1) to (u2);
\draw[thick, ] (u2) to (u3);
\draw[thick, ] (u3) to (u4);
\draw[thick, ] (u4) to (u5);
\draw[thick, ] (u5) to (u6);

\def\xdisp{4.5}

\node[] at (0+\xdisp,2+2.7*\ystep) {\normalsize Mazurkiewizc-based};
\node[state, minimum width=2.5cm] (x1) at (0+\xdisp,2+1*\ystep) {$\Event'_2: \mathsf{balance}\gets 4$};
\node[state, minimum width=2.5cm] (x2) at (0+\xdisp,2+2*\ystep) {$\Event'_1: \mathsf{release}~l$};

\begin{scope}[on background layer]
\draw[thick, ] (x2) to (x1);
\draw[thick,] (x1) to  (v1);
\draw[thick,] (x1) to  (u1);
\end{scope}

\draw[thick, dashed, bend right=\bend ] (v4) to (u4);
\draw[thick, dashed, bend right=\bend] (u4) to (v4);


\draw[thick, dashed, bend right=\bend ] (v4) to (u2);
\draw[thick, dashed, bend left=\bend] (u4) to (v2);
\draw[thick, dashed, bend left=\bend ] (v2) to (u4);
\draw[thick, dashed, bend right=\bend] (u2) to (v4);

\draw[thick, dashed, bend right=\bend ] (v4) to (u6);
\draw[thick, dashed, bend left=\bend] (u4) to (v6);
\draw[thick, dashed, bend left=\bend ] (v6) to (u4);
\draw[thick, dashed, bend right=\bend] (u6) to (v4);


\def\xdisp{7}

\def\bend{20}

\node[] at (0+\xdisp,2+0.75*\ystep) {\large $\Process_1$};
\node[state] (a1) at (0+\xdisp,2-0*\ystep) {$\Event_1$};
\node[state] (a2) at (0+\xdisp,2-1*\ystep) {$\Event_2$};
\node[state] (a3) at (0+\xdisp,2-2*\ystep) {$\Event_3$};
\node[state] (a4) at (0+\xdisp, 2-3*\ystep) {$\Event_4$};
\node[state] (a5) at (0+\xdisp,2-4*\ystep) {$\Event_5$};
\node[state] (a6) at (0+\xdisp,2-5*\ystep) {$\Event_6$};
\draw[thick, ] (a1) to (a2);
\draw[thick, ] (a2) to (a3);
\draw[thick, ] (a3) to (a4);
\draw[thick, ] (a4) to (a5);
\draw[thick, ] (a5) to (a6);

\def\xdisp{10}

\node[] at (0+\xdisp,2+0.75*\ystep) {\large $\Process_2$};
\node[state] (b1) at (0+\xdisp,2-0*\ystep) {$\Event_1$};
\node[state] (b2) at (0+\xdisp,2-1*\ystep) {$\Event_2$};
\node[state] (b3) at (0+\xdisp,2-2*\ystep) {$\Event_3$};
\node[state] (b4) at (0+\xdisp, 2-3*\ystep) {$\Event_4$};
\node[state] (b5) at (0+\xdisp,2-4*\ystep) {$\Event_5$};
\node[state] (b6) at (0+\xdisp,2-5*\ystep) {$\Event_6$};
\draw[thick, ] (b1) to (b2);
\draw[thick, ] (b2) to (b3);
\draw[thick, ] (b3) to (b4);
\draw[thick, ] (b4) to (b5);
\draw[thick, ] (b5) to (b6);

\def\xdisp{8.5}

\node[] at (0+\xdisp,2+2.7*\ystep) {\normalsize Observation-based};
\node[state, minimum width=2.5cm] (y1) at (0+\xdisp,2+1*\ystep) {$\Event'_2: \mathsf{balance}\gets 4$};
\node[state, minimum width=2.5cm] (y2) at (0+\xdisp,2+2*\ystep) {$\Event'_1: \mathsf{release}~l$};

\def\bend{40}
\begin{scope}[on background layer]
\draw[thick, ] (y2) to (y1);
\draw[thick,] (y1) to  (a1);
\draw[thick,] (y1) to  (b1);
\draw[thick, <-, dashed, bend left=\bend] (y1) to (a2);
\draw[thick, <-,dashed, bend right=\bend] (y1) to (b2);
\draw[thick, <-,dashed, ] (a4) to (b2);
\draw[thick, <-,dashed, ] (b4) to (a2);
\draw[thick, <-,dashed, ] (a4) to (b6);
\draw[thick, <-,dashed, ] (b4) to (a6);
\draw[thick, <-,dashed, bend left=\bend] (a4) to (a6);
\draw[thick, <-,dashed, bend right=\bend] (b4) to (b6);
\end{scope}

\end{tikzpicture}
\caption{
Trace exploration on the system of \cref{fig:model_example} with two processes, where initially $\mathsf{balance}\gets 4$ and both withdrawals succeed.
}
\label{fig:m_vs_obs}
\end{figure*}
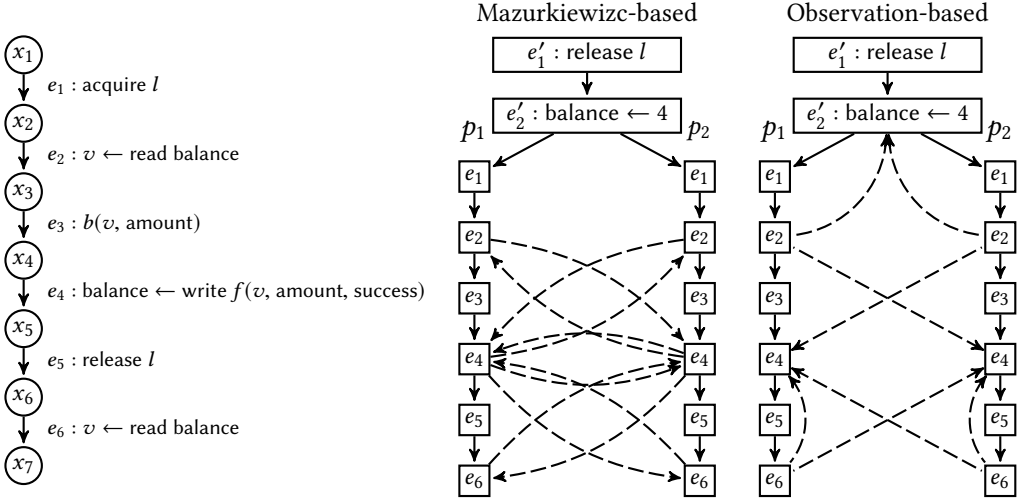

Figure~\ref{fig:m_vs_obs} illustrates the difference between the Mazurkiewicz and observation trace equivalence on the example of Figure~\ref{fig:model_example}.
Every execution of the system starts with an initialization trace $\TraceInit$ that initializes the lock $l$ to $\False$,
and the initial value $\mathsf{desposit}=4$.
Consider that $\Process_1$ is executed with parameter $\mathsf{amount}=1$ and
$\Process_2$ is executed with parameter $\mathsf{amount}=2$, 
(hence both withdrawals succeed).
The primed events $\Event'_1,\Event'_2$ represent the system initialization.
\begin{compactitem}
\item \textit{(Left)}:~The sequential trace of $\Process_1,\Process_2$.
\item \textit{(Center)}:~Trace exploration using the Mazurkiewicz equivalence $\sim_{M}$.
Solid lines represent the happens-before relation enforced by the program structure.
Dashed lines represent potential happens-before relations between dependent events.
A control-centric DPOR based on $\sim_{M}$ will resolve scheduling choices by exploring all possible realizable sets of the happens-before edges.
\item \textit{(Right)}:~Trace exploration using the observation equivalence $\sim_{\Obs}$.
Solid lines represent the happens-before relation enforced by the program structure.
This time, dashed lines represent potential observation functions.
Our data-centric DPOR based on $\sim_{\Obs}$ will resolve scheduling choices by exploring all possible realizable sets of the observation edges.
\end{compactitem}
Both methods are guaranteed to visit all local states of each process.
However, the data-centric DPOR achieves this by exploring potentially fewer scheduling choices.
\end{example}


\subsection{Exponential Succinctness}\label{subsec:exponential_succinct}
As we have already seen in the example of Figure~\ref{fig:m_comparison_intro}, Theorem~\ref{them:comparison} does not hold in the other direction,
i.e., $\sim_{\Obs}$ can be strictly coarser than $\sim_{M}$.
Here we provide two simple examples in which $\sim_{\Obs}$ is exponentially more succinct than $\sim_{M}$.
Traditional enumerative model checking methods of concurrent systems are based on exploring 
{\em at least} one trace from every partition of the Mazurkiewicz equivalence using POR techniques that prune away equivalent traces (e.g. sleep sets~\cite{G96}, persistent sets~\cite{Flanagan05}, source sets and wakeup trees~\cite{Abdulla14}).
Such a search is \emph{optimal} if it explores at most one trace from each class.
Any optimal enumerative exploration based on the observation equivalence is guaranteed by Theorem~\ref{them:comparison} to examine no more traces than any enumerative exploration based on the Mazurkiewicz equivalence.
The two examples show $\sim_{\Obs}$ can offer exponential improvements wrt two parameters: (i)~the number of processes, and (ii)~the size of each process.

\begin{example}[Two processes of large size]
Consider the system $\System$ of $k=2$ processes of Figure~\ref{fig:m_comparison2},
and for $i\in\{1,\dots,n\}, j\in\{1,2\}$, denote by $\Write_i^j$ (resp. $\Read^j$) the $i$-th write event (resp. the read event) of $\Process_j$.
In any maximal trace, there are two ways to order the read events $\Read^1,\Read^2$, i.e., $\Read^j$ occurs before $\Read^{3-j}$ for the two choices of $j\in\{1,2\}$. In any such ordering, $\Read^{3-j}$ can only observe either $\Write^{3-j}_{n-1}$ or $\Write^{j}_{n-1}$, whereas there are at most $n+1$ possible write events for $\Read^j$ to observe (either $\Write^j_{n}$ or one of the $\Write^{3-j}_i$).
Hence $\TraceSpaceMax/\sim_{\Obs}$ has size $O(n)$.
In contrast, $\TraceSpaceMax/\sim_{M}$ has size $\Omega(\binom{2 \cdot n}{n})=\Omega(2^n)$,
as there are $(2\cdot n)!$ ways to order the $2\cdot n$ write events of the two processes,
but $n!\cdot n!$ orderings are invalid as they violate the program structure.
Hence, even for only two processes, the observation equivalence reduces the number of partitions from exponential to linear.
\begin{figure}[]
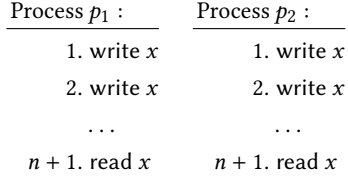

\small
\centering
\begin{subfigure}[t]{0.15\textwidth}
\begin{align*}
\text{Process}&~\Process_1:\\
\hline\\[-1em]
1.&~\mathsf{write}~x\\
2.&~\mathsf{write}~x\\
& \dots\\
n+1.&~\mathsf{read}~x
\end{align*}
\end{subfigure}
\quad
\begin{subfigure}[t]{0.15\textwidth}
\begin{align*}
\text{Process}&~\Process_2:\\
\hline\\[-1em]
1.&~\mathsf{write}~x\\
2.&~\mathsf{write}~x\\
& \dots\\
n+1.&~\mathsf{read}~x
\end{align*}
\end{subfigure}
\caption{An architecture of two processes with $n+1$ events each.}
\label{fig:m_comparison2}
\vspace{-1em}
\end{figure}

\end{example}

\begin{example}[Many processes of small size]
We now turn our attention to a system $\System$ of $k$ identical processes $\Process_1,\dots,\Process_k$ with two events each, in Figure~\ref{fig:m_comparison1}.
There is only one global variable $x$, and each process performs a read and then a write to $x$.
There are $O(k^k)$ realizable observation functions, by choosing for each one among $k$ read events, one among $k$ write events it can observe. Hence $\TraceSpaceMax/\sim_{\Obs}$ has size $O(k^k)$.
In contrast, the size of $\TraceSpaceMax/\sim_{M}$ is $\Omega((k!)^2)$.
This holds as there are $k!$ ways to order the $k$ write events, and for each such permutation there are $k!$ ways to assign each of the $k$ read events to the write event that it observes.
To see this second part, let $\Write_1,\dots,\Write_k$ be any permutation of the write events, and let $\Read_i$ be the read event in the same process as $\Write_i$. Then $\Read_i$ can be placed right after any $\Write_j$ with $i\leq j$.
Observe that $\TraceSpaceMax/\sim_{\Obs}$ is exponentially more succinct than $\TraceSpaceMax/\sim_{M}$, as 
\[
\frac{\Omega((k!)^2)}{O(k^k)} = \Omega\left(\frac{\prod_{i=1}^k i\cdot\lceil\frac{k}{i}\rceil }{k^k}\cdot \prod_{i=\lceil\frac{k}{2}\rceil+1}^{k-1}i\right)= \Omega(2^k).
\]
\begin{figure}[]
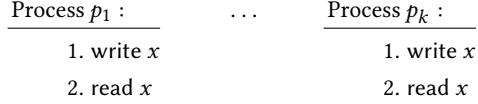

\small
\centering
\begin{subfigure}[t]{0.15\textwidth}
\begin{align*}
\text{Process}&~\Process_1:\\
\hline\\[-1em]
1.&~\mathsf{write}~x\\
2.&~\mathsf{read}~x
\end{align*}
\end{subfigure}
\qquad
\begin{subfigure}[t]{0.05\textwidth}
\begin{align*}
\dots
\end{align*}
\end{subfigure}
\qquad
\begin{subfigure}[t]{0.15\textwidth}
\begin{align*}
\text{Process}&~\Process_k:\\
\hline\\[-1em]
1.&~\mathsf{write}~x\\
2.&~\mathsf{read}~x
\end{align*}
\end{subfigure}
\caption{An architecture of $k$ processes with two events each.}
\label{fig:m_comparison1}
\vspace{-1em}
\end{figure}

\end{example}

\subsection{Solution Overview}
Traditional DPOR algorithms exploit the Mazurkiewicz equivalence, and use various techniques
such as persistent sets and sleep sets to explore each Mazurkiewicz class by few representative traces.
Our goal is to develop an analogous DPOR that utilizes the observation equivalence,
which by Theorem~\ref{them:comparison} is more succinct.
In high level, our approach consists of the following steps.
\begin{compactenum}
\item In Section~\ref{sec:annotations} we introduce the concept of annotations.
An annotation is a function from read to write events, and serves as an intended observation function.
Given an annotation, the goal is to obtain a trace whose observation function coincides with the annotation.
We restrict our attention to a certain class of \emph{well-formed} annotations, and show that although the problem is \NPC~in general,
it admits a polynomial time (in fact, cubic in the size of the trace) solution in acyclic architectures.
\item In Section~\ref{sec:enumerative} we present our data-centric DPOR. 
Section~\ref{sec:cones} introduces the notion of causal past cones in a trace.
The concept is similar to Lamport's \emph{happens-before} relation~\cite{Lamport78},
and is used to identify past events that may causally affect a current event in a trace.
We note that this concept is different from the happens-before relation used in the Mazurkiewicz equivalence.
We use the notions of annotations and causal cones to develop our algorithm, and prove its correctness and optimality 
(in Section~\ref{subsec:dcdpor}).
\item In Section~\ref{sec:cyclic_architectures} we extend our algorithm to cyclic architectures.
\end{compactenum}

Table~\ref{tab:not1} and Table~\ref{tab:not2} summarize relevant notation in the proofs.

\section{Annotations}\label{sec:annotations}

In this section we introduce the notion of \emph{annotations}, which are 
intended constraints on the observation functions that traces discovered by 
our data-centric DPOR ($\EnumExplore$) are required to meet.

\noindent{\bf Annotations.}
An \emph{annotation pair} $\Annotation=(\Annotationpos, \Annotationneg)$ is a pair of
\begin{compactenum}
\item a \emph{positive annotation} $\Annotationpos: \SysReads \pto \SysWrites$, and
\item a \emph{negative annotation} $\Annotationneg: \SysReads \pto 2^{\SysWrites}$
\end{compactenum}
such that for all read events $\Read$, if $\Annotationpos(\Read)=\Write$, then we have $\Confl(\Read, \Write)$ and it is not the case that $\PS(\Read, \Write)$.
We will use annotations to guide the recursive calls of $\EnumExplore$ towards traces that belong to different equivalence classes
than the ones explored already, or will be explored by other branches of the algorithm.
A positive annotation $\Annotationpos$ forces $\EnumExplore$ to explore traces that are compatible with $\Annotationpos$
(or abort the search if no such trace can be generated).
Since a positive annotation is an ``intended'' observation function, we say that a trace $\Trace$ \emph{realizes} $\Annotationpos$ if $\Obs_{\Trace}=\Annotationpos$, in which case $\Annotationpos$ is called \emph{realizable}.
A negative annotation $\Annotationneg$ prevents $\EnumExplore$ from exploring traces $\Trace$ in which a read event observes a write event that belongs to its negative annotation set (i.e., $\Obs_{\Trace}(\Read)\in \Annotationneg(\Read)$).
In the remaining section we focus on positive annotations, and the problem of deciding whether a positive annotation is realizable.

\noindent{\bf The value function $\Value_{\Annotationpos}$.}
Given a positive annotation $\Annotationpos$, we define the relation $<_{\Annotationpos}\subseteq \Image(\Annotationpos)\times \Domain(\Annotationpos)$ such that $\Write <_{\Annotationpos} \Read$ iff $(\Read, \Write)\in \Annotationpos$.
The positive annotation $\Annotationpos$ is \emph{acyclic} if the relation $\PS \cup <_{\Annotationpos}$ is a strict partial order (i.e., it contains no cycles).
The \emph{value function} $\Value_{\Annotationpos}:\Domain(\Annotationpos)\cup \Image(\Annotationpos)\to \ValueDomain$ of an acyclic positive annotation $\Annotationpos$ is the unique function defined inductively, as follows.
\begin{compactenum}
\item For each $\Write\in\Image(\Annotationpos)$ of the form $\Write:g\gets \mathsf{write}~f(v_1,\dots,v_{n_i})$, we have $\Value_{\Annotationpos}(\Write)=f(\alpha_1,\dots, \alpha_{n_i})$, where for each $\alpha_j$ we have
\begin{compactenum}
\item $\alpha_j=\Value_{\Annotationpos}(\Read)$ if there exists a read event $\Read\in \Domain(\Annotationpos)$ 
such that (i)~$\Read$ is of the form $\Read: v_j\gets \mathsf{read}~g'$ and (ii)~$\PS(\Read, \Write)$ and (iii)~there exists no other $\Read'\in \Domain(\Annotationpos)$ with $\PS(\Read, \Read')$ and which satisfies conditions (i)~and (ii).
\item $\alpha_j$ equals the initial value of $v_i$ otherwise.
\end{compactenum}
\item For each $\Read\in \Domain(\Annotationpos)$ we have $\Value_{\Annotationpos}(\Read)= \Value_{\Annotationpos}(\Annotationpos(\Read))$.
\end{compactenum}
Note that $\Value_{\Annotationpos}$ is well-defined, as for any read event $\Read$ that is used to define the value of a write event $\Write$ we have $\PS(\Read, \Write)$, and thus by the acyclicity of $\Annotationpos$, $\Value_{\Annotationpos}(\Read)$ does not depend on $\Value_{\Annotationpos}(\Write)$.

\begin{remark}\label{rem:annotationpos_value}
If $\Annotationpos$ is realizable then it is acyclic,
and for any trace $\Trace$ that realizes $\Annotationpos$ we have that $\Value_{\Trace}=\Value_{\Annotationpos}$.
\end{remark}

\noindent{\bf Well-formed annotations and basis of annotations.}
A positive annotation $\Annotationpos$ is called \emph{well-formed} if the following conditions hold:
\begin{compactenum}
\item\label{item:acyclicity} $\Annotationpos$ is acyclic.
\item\label{item:lock_obs} For every lock-release event $\Event_{a}\in \Image(\Annotationpos)\cap \SysAcquires$
there is at most one lock-acquire event $\Event_{r}\cap \SysReleases$ such that $\Annotationpos(\Event_{a})=\Event_{r}$.
\item\label{item:basis} There exist sequential traces $(\SeqTrace_i)_i$, one for each process $\Process_i$, such that each $\SeqTrace_i$ ends in a global event, and the following conditions hold. 
\begin{compactenum}
\item for every pair of lock-acquire events $\Event_a^1, \Event_a^2\in \Events{\SeqTrace_i}\cap \SysAcquires$ such that $\Index_{\SeqTrace_i}(\Event_a^1)<\Index_{\SeqTrace_i}(\Event_a^2)$ and $\Location(\Event_a^1)=\Location(\Event_a^2)$
there exists a lock release event $\Event_r\in \Events{\SeqTrace_i}\cap \SysReleases$ such that 
$\Index_{\SeqTrace_i}(\Event_a^1)<\Index{\SeqTrace_i}(\Event_r)< \Index{\SeqTrace_i}(\Event_a^2)$ and $\Location(\Event_r)=\Location(\Event_a^1)=\Location(\Event_a^2)$.
\item $\bigcup_i \Reads{\SeqTrace_i} = \Domain(\Annotationpos)$ and $\Image(\Annotationpos) \subseteq \bigcup_i \Writes{\SeqTrace_i}$,
i.e., $(\SeqTrace_i)_i$ contains precisely the read events of $\Annotationpos$ and a superset of the write events.
\item Each $\SeqTrace_i$ corresponds to a deterministic computation of process $\Process_i$, where the value of every global event $\Event$ during the computation is taken to be $\Value_{\Annotationpos}(\Event)$.
\end{compactenum}
\end{compactenum}

The sequential traces $(\SeqTrace_i)_i$ are called a \emph{basis} of $\Annotationpos$ if every $\SeqTrace_i$ is minimal in length.
The following lemma establishes properties of well-formedness and basis.

\begin{restatable}{lemma}{annotationbasis}
\label{lem:annotation_basis}
Let $X=\Domain(\Annotationpos)\cup \Image(\Annotationpos)$ be the set of events that appear in a positive annotation $\Annotationpos$, and $X_i=X\cap \SysEvents_i$ the subset of events of $X$ from process $\Process_i$.
The following assertions hold:
\begin{compactenum}
\item If $\Annotationpos$ is well-formed, then it has a unique basis $(\SeqTrace_i)_i$.

\item Computing the basis of $\Annotationpos$ (or concluding that $\Annotationpos$ is not well-formed)
can be done in $O(n)$ time, where $n=\sum_i(|\SeqTrace_i|)$
if $\Annotationpos$ is well-formed, otherwise $n= \sum_i \ell_i$,
where $\ell_i$ is the length of the longest path from the root $\Root_i$ of $\CFG_i$ to an event $\Event\in X_i$.

\item For every trace $\Trace$ that realizes $\Annotationpos$ we have that 
$\Annotationpos$ is well-formed and $\Trace\in \SeqTrace_1 \Conv \dots \Conv \SeqTrace_k$.
\end{compactenum}
\end{restatable}

%
%
%

\subsection{The Hardness of Realizing Positive Annotations}\label{subsec:get_trace}

A core step in our data-centric DPOR algorithm is constructing a trace that realizes a positive annotation.
That is, given a positive annotation $\Annotationpos$, the goal is to obtain a trace $\Trace$ (if one exists) such that
$\Obs_{\Trace}=\Annotationpos$, i.e., $\Trace$ contains precisely the read events of $\Annotationpos$,
and every read event in $\Trace$ observes the write event specified by $\Annotationpos$.
Here, we show that the problem is \NPC~in the general case.
Membership in \NP~is trivial, since, given a trace $\Trace$, it is straightforward to verify that $\Obs_{\Trace}=\Annotationpos$ in $O(|\Trace|)$ time.
Hence our focus will be on establishing \NP-hardness.
For doing so, we introduce a new graph problem, namely \AEA, which is closely related to the problem of realizing a positive annotation under sequential consistency semantics.
We first show that \AEA~is \NPH, and afterwards that the problem is polynomial-time reducible to realizing a positive annotation.

\noindent{\bf The problem \AEA.}
The input to the problem is a pair $(G,H)$ where $G=(V,E)$ is a directed acyclic graph, and $H=\{(x_i,y_i, z_i)\}_i$ is a set of triplets of pairwise distinct nodes such that
\begin{compactenum}
\item $x_i,y_i,z_i\in V$, $(x_i,y_i)\in E$, and
\item each node $x_i$ and $y_i$ appears in a unique triplet of $H$.
\end{compactenum}
An \emph{edge addition set} $X=\{e_i\}_{i=1}^{|H|}$ for $(G,H)$ is a set of edges $e_i\in E$ such that for each $e_i$ we have either $e_i=(z_i,x_i)$ or $e_i=(y_i,z_i)$.
The problem \AEA~asks whether there exists an edge addition set $X$ for $(G,H)$ such that the graph $G_X=(V, E\cup X)$ remains acyclic.
The problem \UAEA~is similar to \AEA, with the restriction that every node $z_i$ appears in a unique triplet.

\begin{restatable}{lemma}{aea}
\label{lem:aea}
\AEA~is \NPH.
\end{restatable}
\noindent{\em Sketch.}
The proof is by reduction from \SATMONOTONE~\cite[LO4]{Garey79}.
In \SATMONOTONE, the input is a propositional 3CNF formula $\phi$ in which every literal is positive,
and the goal is to decide whether there exists a satisfying assignment for $\phi$ that assigns exactly one literal per clause to $\True$.
The reduction proceeds as follows.
In the following, we let $C$ and $D$ range over the clauses and $x_i$ over the variables of $\phi$.
We assume w.l.o.g. that no variable repeats in the same clause.
For every variable $x_i$, we introduce a node $\Write'_i\in V$.
For every clause $C=(x_{C_1}\lor x_{C_2} \lor x_{C_3})$, we introduce a pair of nodes $\Write^C_{C_j}, \Read^C_{C_j}\in V$
and an edge $(\Write^C_{C_j}, \Read^C_{C_j})\in E$, where $j\in \{1,2,3\}$.
Additionally, we introduce an edge $(\Write^C_{C_j}, \Write'_{C_l})\in E$ for every pair $j,l\in\{1,2,3\}$ such that $j\neq l$,
and an edge $(\Write'_{C_j}, \Read^C_{C_l})$ for each $j\in \{1,2,3\}$, where $l=( j+1) \mod 3 + 1$.
Finally, for every pair of clauses $C, D$ and $l_1,l_2\in \{1,2,3\}$ such that $C_{l_1}=D_{l_2}=\ell$
(i.e., $C$ and $D$ share the same variable $x_{\ell}$ in positions $l_1$ and $l_2$),
we add edges $(\Write^{C}_{\ell}, \Read^{D}_{\ell}), (\Write^{D}_{\ell}, \Read^{C}_{\ell})\in E$.
The set $H$ consists of triplets of nodes $(\Write^C_{C_j},\Read^C_{C_j},\Write'_{C_j})$ for every clause $C$ and $j\in\{1,2,3\}$.
Figure~\ref{fig:reduction} illustrates the construction.

\begin{figure}[!t]
\removelatexerror
\centering
\newcommand{\distone}{1cm}
\footnotesize
\centering
\begin{tikzpicture}[->,>=stealth',shorten >=-2pt,shorten <=-2pt,auto,node distance=\distone,
                    semithick,scale=1 ]
      
\tikzstyle{every state}=[fill=white,draw=white,text=black,font=\small, thick, rectangle, inner sep=0.05cm, minimum size=0.4cm]
\tikzstyle{invis}=[fill=white,draw=white,text=white,font=\small , inner sep=-0.05cm]
\tikzstyle{dashed}=                  [dash pattern=on 6pt off 2pt]

\def\xstep{3}
\def\ydispl{-1}
\def\ydispg{-2.3}
\def\bendone{20}
\def\bendtwo{30}

\node[state, circle] (wC1) at (0*\xstep,0*\ydispg + 0*\ydispl) {$\Write^C_1$};
\node[state, circle] (rC1) at (0*\xstep,0*\ydispg + 1*\ydispl) {$\Read^C_1$};

\node[state, circle] (wC2) at (1*\xstep,0*\ydispg + 0*\ydispl) {$\Write^C_2$};
\node[state, circle] (rC2) at (1*\xstep,0*\ydispg + 1*\ydispl) {$\Read^C_2$};

\node[state, circle] (wC3) at (2*\xstep,0*\ydispg + 0*\ydispl) {$\Write^C_3$};
\node[state, circle] (rC3) at (2*\xstep,0*\ydispg + 1*\ydispl) {$\Read^C_3$};

\node[state, circle] (wD1) at (0*\xstep,1*\ydispg + 0*\ydispl) {$\Write^D_1$};
\node[state, circle] (rD1) at (0*\xstep,1*\ydispg + 1*\ydispl) {$\Read^D_1$};

\node[state, circle] (wD4) at (1*\xstep,1*\ydispg + 0*\ydispl) {$\Write^D_4$};
\node[state, circle] (rD4) at (1*\xstep,1*\ydispg + 1*\ydispl) {$\Read^D_4$};

\node[state, circle] (wD5) at (2*\xstep,1*\ydispg + 0*\ydispl) {$\Write^D_5$};
\node[state, circle] (rD5) at (2*\xstep,1*\ydispg + 1*\ydispl) {$\Read^D_5$};

\draw[thick] (wC1) to  (rC1);
\draw[thick] (wC2) to  (rC2);
\draw[thick] (wC3) to  (rC3);

\draw[thick] (wD1) to  (rD1);
\draw[thick] (wD4) to  (rD4);
\draw[thick] (wD5) to  (rD5);

\draw[thick, bend right=\bendone] (wC1) to  (rD1);
\draw[thick, bend right=0] (wD1) to  (rC1);

\node[state, circle] (w1) at (0.5*\xstep,0*\ydispg + 0.4*\ydispl) {$\Write'_1$};
\node[state, circle] (w2) at (1.5*\xstep,0*\ydispg + 0.4*\ydispl) {$\Write'_2$};
\node[state, circle] (w3) at (2.5*\xstep,0*\ydispg + 0.4*\ydispl) {$\Write'_3$};

\node[state, circle] (w4) at (1.5*\xstep,1*\ydispg + 0.4*\ydispl) {$\Write'_4$};
\node[state, circle] (w5) at (2.5*\xstep,1*\ydispg + 0.4*\ydispl) {$\Write'_5$};

\draw[thick, bend left=\bendtwo] (wC1) to  (w2);
\draw[thick, bend left=\bendtwo] (wC1) to  (w3);
\draw[thick, bend left=\bendtwo] (wC2) to  (w3);
\draw[thick, bend left=0] (wC2) to  (w1);
\draw[thick, bend left=0] (wC3) to  (w2);
\draw[thick, bend right=\bendtwo] (wC3) to  (w1);

\draw[thick, bend left=\bendtwo] (wD1) to  (w4);
\draw[thick, bend left=\bendtwo] (wD1) to  (w5);
\draw[thick, bend left=\bendtwo] (wD4) to  (w5);
\draw[thick, bend left=0] (wD4) to  (w1);
\draw[thick, bend left=0] (wD5) to  (w4);
\draw[thick, bend left=\bendone] (wD5) to  (w1);

\draw[thick, bend left=0] (w1) to  (rC2);
\draw[thick, bend left=0] (w1) to  (rD4);
\draw[thick, bend left=0] (w2) to  (rC3);
\draw[thick, bend left=0] (w4) to  (rD5);

\draw[thick, rounded corners] (w3) to (2.5*\xstep,0*\ydispg + -0.3*\ydispl) to [bend right=20] (0*\xstep,0*\ydispg + -0.3*\ydispl) to [bend right=40] (rC1);

\draw[thick, rounded corners] (w5) to (2.5*\xstep,1*\ydispg + 1*\ydispl)  to [bend left=15] (rD1);

\draw[thick, dashed, color=black] (w1) to  (wC1);
\draw[thick, dashed, color=black]  (rC1) to (w1);
\draw[thick, dashed, color=black] (w2) to  (wC2);
\draw[thick, dashed, color=black]  (rC2) to (w2);
\draw[thick, dashed, color=black] (w3) to  (wC3);
\draw[thick, dashed, color=black]  (rC3) to (w3);

\draw[thick, dashed, color=black] (w1) to  (wD1);
\draw[thick, dashed, color=black]  (rD1) to (w1);
\draw[thick, dashed, color=black] (w4) to  (wD4);
\draw[thick, dashed, color=black]  (rD4) to (w4);
\draw[thick, dashed, color=black] (w5) to  (wD5);
\draw[thick, dashed, color=black]  (rD5) to (w5);

\node[] at (3.5, 1.6)	{$\phi=\underbrace{(x_1\lor x_2 \lor x_3)}_{C} \land \underbrace{(x_1\lor x_4 \lor x_5)}_{D}$};

\end{tikzpicture}
\caption{The reduction of 3SAT over $\phi$ to \AEA~over $(G,H)$. The nodes and solid edges represent the graph $G$. The dashed edges represent the triplets in $H$.}
\label{fig:reduction}
\end{figure}
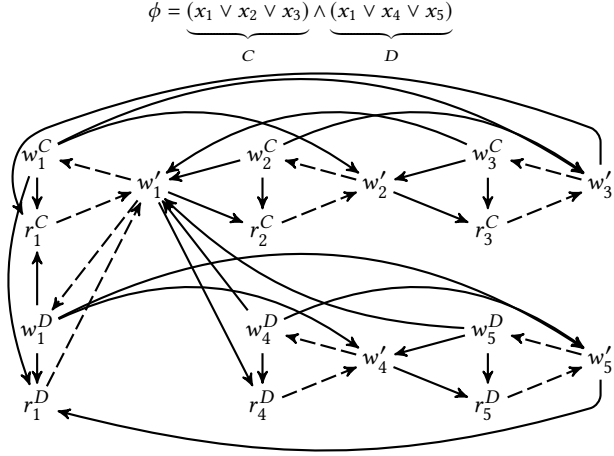

\noindent{\bf From \AEA~to \UAEA.}
Here we show that \UAEA~is \NPH, by a reduction from \AEA.
Let $(G=(V,E),H)$ be an instance of \AEA, fix a total order on the triplets of $H$, and denote by $T_i=(x_i, y_i, z_i)$ the $i$-th triplet of $H$.
We call two triplets $(x_i, y_i, z_i)$ and $(x_{i'}, y_{i'}, z_{i'})$ related if $z_i=z_{i'}$.
For simplicity, we make the following assumptions.
\begin{compactenum}
\item $G$ is transitively closed, as a graph has a cycle iff its transitive closure has a cycle. 
\item The set of nodes $V$ is precisely the set of nodes that appear in the triplets of $H$.
Indeed, any other node can be removed while maintaining the connectivity between the nodes in the triplets of $H$,
and any edge addition set solves the problem in the original graph iff it does so in the reduced graph.
\item If the $i$-th and $i'$-th triplets of $H$ are related, then $(x_i, y_{i'}), (x_{i'}, y_i)\in E$.
This is sound as the instances we construct in our reduction from \SATMONOTONE~ to \AEA~have this property.
\end{compactenum}

We now proceed with the reduction.
We construct at instance $(G'=(V',E'),H')$ of \UAEA, as follows.
\begin{compactenum}
\item For every triplet $T_i=(x_i, y_i, z_i)$ of $H$, we have $x_i, y_i, z_i^j\in V'$ and $(x_i, y_i)\in E'$, where $j$ equals one plus the number of triplets $T_{i'}=(x_{i'}, y_{i'}, z_{i'})$ that are related to $T_i$ and such that $i'< i$.
We add a triplet $(x_i, y_i, z_i^j)\in H'$.
If $j>1$, we also add $(z_{i'}^1, z_{i}^j)\in E'$, for the appropriate choice of $i'$.
\item For every triplet $T_{i}=(x_{i}, y_{i}, z_{i}^j)$ with $j>1$, we introduce nodes $a_{i}, b_{i}, c_{i}, d_{i}, e_{i}, f_{i}\in V'$, and triplets
$(a_{i}, b_{i}, c_{i}), (d_{i}, e_{i}, f_{i})\in H'$.
Let ${i'}$ be such that $(x_{i'}, y_{i'}, z_{i'}^1)$ is a triplet already in $H'$.
We add the edges $(a_{i}, x_{i'}), (y_{i'}, f_{i}), (f_{i}, b_{i})\in E'$
and $(d_{i}, z_{i'}^{1}), (z_{i}, c_{i}), (c_{i}, e_{i})\in E'$.
\end{compactenum}
Note that $(G',H')$ is polynomial in the size of $(G,H)$.
We now proceed with the correctness of the reduction.
If a pair of triplets $(x_i, y_i, z_i)$ and $(x_{i'}, y_{i'}, z_{i'})$ are related in $H$, we say that the corresponding pair $(x_i, y_i, z_i^j)$ and $(x_{i'}, y_{i'}, z_{i'}^{j'})$ of triplets in $H'$, for appropriate $j, j'$, are related.
First, we make some key observations.
\begin{compactenum}
\item\label{item:obs5} If $X$ is an edge addition set for $(G, H)$ then for every pair of related triplets $(x_i, y_i, z_i)$ and $(x_{i'}, y_{i'}, z_{i'})$ we have $(z_i, x_i)\in X$ iff $(z_{i'}, x_{i'})\in X$.
\item\label{item:obs1} For every $j>1$ and related triplets $(x_i, y_i, z_i^j)$, $(x_{i'}, y_{i'}, z_{i'}^1)$, for appropriate $i$, $i'$ and $j'$,
adding an edge $(z_{i'}^1, x_{i'})$ in $G'$ leads to $d_{i}\Path f_{i}$.
\item\label{item:obs2} For every $j>1$ and triplet $(x_i, y_i, z_i^j)$, for appropriate $i$, adding an edge $(e_{i}, f_{i})$ in $G'$ leads to $c_{i}\Path b_{i}$.
\end{compactenum}

\noindent{\em \AEA $\implies$ \UAEA.}
Consider an edge addition set $X$ for $(G,H)$.
We construct an  edge addition set $X'$ for $(G', H')$ as follows.
Given a triplet $(x_i, y_i, z_i)$ of $H_i$, let $j$ be such that $z_i$ occurs for the $j$-th time in a triplet of $H$.
\begin{compactenum}
\item\label{item:aea_uaea_step1} If $(z_i, x_i)\in X$ then we introduce $(z_i^j, x_i)\in X'$, otherwise we  introduce $(y_i, z_i^j)\in X'$.
\item\label{item:aea_uaea_step2} If $j>1$, if $(z_i^j, x_i)\in X$, we introduce $(e_i, f_i), (c_i, a_i)\in X'$. Otherwise, we introduce $(b_i, c_i), (f_i, d_i)\in X'$.
\end{compactenum}
Due to observation~\ref{item:obs5} above, the edges introduced step~\ref{item:aea_uaea_step1} do not lead to a cycle in $G'$.
It is also easy to see that after having introduced the edges of step~\ref{item:aea_uaea_step1}, introducing each pair of edges in step~\ref{item:aea_uaea_step2}  does not lead to a cycle. Hence $X'$ is an edge addition set for $(G', H')$.

\noindent{\em \UAEA $\implies$ \AEA.}
Consider an edge addition set $X'$ for $(G',H')$.
We construct an edge addition set $X$ for $(G, H)$ as follows.
For every triplet $(x_i, y_i, z_i)\in H$, we have $(z_i, x_i)\in X$ iff $(x_i, z_i^j)\in X'$,
where $j$ equals one plus the number of triplets $T_{i'}=(x_{i'}, y_{i'}, z_{i'}^{j'})$ that are related to $T_i$ and such that $i'< i$.
We now argue that $X$ is an edge addition set for $(G, H)$.
This fact is a consequence of the following observation: if $T_{i'}=(x_{i'}, y_{i'}, z_{i'}^1)$ is a triplet of $H'$, then for every triplet $T_{i}=(x_{i}, y_{i}, z_{i}^{j})$ of $H'$ that is related to $T_i$, we have $(z_{i'}^1, x_{i'})\in X'$ iff $(z_{i}^j, x_{i})\in X'$.
Indeed:
\begin{compactenum}
\item If $(z_{i'}^1, x_{i'})\not \in X'$ then $(y_{i'}, z_{i'}^1)\in X'$. But then $x_{i}\Path z_{i}^j$ and thus $(z_{i}^j,x_{i})\not \in X'$.
\item If $(z_{i'}^1, x_{i'})\in X'$ then by our observation~\ref{item:obs1} above, we have $d_{i}\Path f_{i}$ and thus $(f_{i}, d_{i})\not \in X'$, hence $(e_{i}, f_{i})\in X'$.
Then, by our observation~\ref{item:obs2} above, we have $c_{i}\Path b_{i}$ and thus $(b_{i}, c_{i})\not \in X'$, hence $(c_{i}, a_{i})\in X'$.
Finally, observe that adding $(c_{i}, a_{i})$ in $G'$ leads to $z_{i}^j\Path y_{i}$ and thus $(y_{i}, z_{i}^j)\not \in X'$, hence $(z_{i}^j, x_{i})\in X'$.
\end{compactenum}

\noindent{\bf From \UAEA~ to annotations.}
Finally, we argue that \UAEA~is polynomial-time reducible to realizing a positive annotation.
Given an instance $(G,H)$ of \UAEA, with $G=(V,E)$, we construct an architecture $\System$ of $k=2\cdot |H|$ processes $(\Process_i)_i$, and a positive annotation $\Annotationpos$.
Similarly to the previous step, we assume that $G$ is transitively closed and $V=\bigcup_i \{ x_i, y_i, z_i \}$, where $T_i=(x_i, y_i, z_i)$ ranges over the triplets of $H$.
The construction consists of the following steps.

\begin{compactenum}
\item For every triplet $T_i=(x_i,y_i,z_i)$, we introduce a global variable $g_i$.
We create two events $\Write_{i}\in \SysWrites$, $\Read_{i}\in \SysReads$ in $\Process_i$, and make $\PS(\Write_i, \Read_i)$.
In addition, we create an write event $\Write'_i\in \SysWrites$ in $\Process_{|H|+i}$, and make $\Location(\Write'_{i})=\Location(\Write_{i})=\Location(\Read_{i})=g_i$.
Finally, we introduce $(\Read_i, \Write_i)\in \Annotationpos$.
We associate $x_i$ (resp., $y_i$, $z_i$) with $\Write_i$ (resp., $\Read_i$, $\Write'_i$).
Given a node $u$ introduced in this step, we let $\Event(u)$ denote the event associated with node $u$.

\item For every node $u$ of the above step we introduce a new global variable $g_u$ and a write event $\Write_u\in \SysWrites$ with $\Proc(\Write_u)=\Proc(\Event(u))$ and $\PS(\Event(u), \Write_u)$ and $\Location(\Write_u)=g_u$.
For every edge $(u,v)$ we introduce a read event $\Read_{u,v}\in \SysReads$ with $\Proc(\Read_{u,v})=\Proc(\Event(v))$ and $\PS(\Read_{u,v}, \Event(v))$ and $\Location(\Read_{u,v})=g_u$.
Finally, we introduce $(\Read_{u,v}, \Write_{u})\in \Annotationpos$.
\end{compactenum}

Observe that the above construction is linear in the size of $(G, H)$.
The following lemma states the correctness of the reduction.

\begin{restatable}{lemma}{aatoann}
\label{lem:aatoann}
The decision problem of \UAEA~on input $(G=(V,E),H)$ admits a positive answer iff the positive annotation $\Annotationpos$ is realizable in $\System$.
\end{restatable}
%

\subsection{Realizing Positive Annotations in Acyclic Architectures}

We now turn our attention to a tractable fragment of the positive annotation problem.
Here we show that if $\System$ is an acyclic architecture, then the problem admits a polynomial-time solution
(in fact, cubic in the size of the constructed trace).

\noindent{\bf Intuition.}
The hardness of realizing positive annotations in general architectures comes from transitivity constraints that 
ensure that the resulting happens-before relation is acyclic. 
That is, for every triplet of events $\Event_1, \Event_2, \Event_3$,  deciding that (i)~$\Event_1$ happens before $\Event_2$ and (ii)~$\Event_2$ happens before $\Event_3$ must lead in (iii)~$\Event_1$ happening before $\Event_3$.
In general architectures, such a triplet of events is, in general, \emph{unrelated a-priori}, and hence an algorithm that constructs a trace out of a positive annotation need to make the above decisions (i)-(iii) consistently.
In contrast, acyclic architectures have the property that in every such triplet of events, a pair of them always belongs to the same process and thus is \emph{ordered a-priori} by the program structure. This allows to express transitivity constraints by means of a 2SAT encoding.
For example assume that $\Event_2,\Event_3$ belong to the same process and $\Event_2$ is ordered before $\Event_3$ by the program structure.
Then the transitivity constraints can be simply encoded in a 2SAT clause 
$(x_{\Event_1, \Event_2}\Rightarrow x_{\Event_1, \Event_3})$, 
where $x_{\Event, \Event'}$ is interpreted as a boolean variable indicating that $\Event$ happens before $\Event'$.
Besides transitivity constraints, observe that positive annotation constraints can also be encoded in a 2SAT clause.
That is, every positive annotation constraint $\Annotationpos(\Read)=\Write$ can be encoded in a 2SAT clause 
$(x_{\Write', \Read} \Rightarrow x_{\Write', \Write})$, 
for every write event $\Write'\neq \Write$ that conflicts with $\Write$.

\noindent{\bf Procedure $\Realize$.}
Let $\System$ be an acyclic architecture, and $\Annotationpos$ a positive annotation over $\System$.
We describe a procedure $\Realize(\Annotationpos)$ which returns a trace $\Trace$ that realizes $\Annotationpos$, or $\bot$ if $\Annotationpos$ is not realizable.
The procedure works in two phases.
In the first phase, $\Realize(\Annotationpos)$ uses Lemma~\ref{lem:annotation_basis} to extract a basis $(\SeqTrace_i)_i$ of $\Annotationpos$.
In the second phase, $\Realize(\Annotationpos)$ determines whether the events of $\bigcup_i \Events{\SeqTrace_i}$ can be linearized in a trace $\Trace$ such that $\Obs_{\Trace}=\Annotationpos$.
Informally, the second phase consists of constructing a 2SAT instance over variables $x_{\Event_1,\Event_2}$, where $\Event_1,\Event_2\in \bigcup_i\Events{\SeqTrace_i}$.
Setting $x_{\Event_1,\Event_2}$ to $\True$ corresponds to making $\Event_1$ happen before $\Event_2$ in the witness trace $\Trace$.
The clauses of the 2SAT instance capture four properties that each such ordering needs to meet, namely that
\begin{compactenum}
\item the resulting assignment produces a total order (totality, antisymmetry and transitivity) between all of the events that appear in adjacent processes in the communication graph $G_{\System}$,
\item the produced total order respects the positive annotation, i.e., every write event $\Write'$ that conflicts with an annotated read/write pair $(\Read,\Write)\in \Annotationpos$ must either happen before $\Write$ or after $\Read$, and
\item the produced total order respects the partial order induced by the program structure $\PS$ and the positive annotation $\Annotationpos$.
\end{compactenum}

\newsavebox{\constructgraph}
\begin{algorithm}
\small
\SetAlgoNoLine
\DontPrintSemicolon
\caption{$\Realize(\Annotationpos)$}\label{algo:realize}
\KwIn{A positive annotation $\Annotationpos$ with basis $(\SeqTrace_i)_i$ }
\KwOut{A trace $\Trace$ that realizes $\Annotationpos$ or $\bot$ if $\Annotationpos$ is not realizable}
\BlankLine
Construct a directed graph $G=(V,E)$ where\label{line:graph_G}\\
\Indp
- $V=\bigcup_i \Events{\SeqTrace_i}$, and\\
- $E=\{(\Event_1, \Event_2):~  (\Event_2,\Event_1)\in \Annotationpos \text{ ~or~ } \PS(\Event_1,\Event_2) \}$\\
\Indm
$G^*=(V, E^*)\gets$ the transitive closure of $G$\label{line:closure}\\
\tcp{A set $\mathcal{C}$ of 2SAT clauses over variables $V_\mathcal{C}$}
$\mathcal{C}\gets \emptyset$\\
$V_{\mathcal{C}}\gets \{x_{\Event_1, \Event_2}:~\Event_1, \Event_2\in V \text{ and } \Event_1\neq \Event_2 \text{ and either } \Proc(\Event_1)=\Proc(\Event_2) \text{ or } (\Proc(\Event_1), \Proc(\Event_2))\in E_{\System} \}$\label{line:variables}\\
\tcp{1.~Antisymmetry clauses}
\ForEach{$x_{\Event_1, \Event_2} \in V_{\mathcal{C}}$}{\label{line:antisymmetry}
$\mathcal{C}\gets \mathcal{C} \cup \{ (x_{\Event_1,\Event_2}\Rightarrow \neg x_{\Event_2, \Event_1}), (\neg x_{\Event_2, \Event_1} \Rightarrow x_{\Event_1, \Event_2}) \}$
}
\tcp{2.~Transitivity clauses}
\ForEach{$x_{\Event_1, \Event_2} \in V_{\mathcal{C}}$}{\label{line:transitivity}
\ForEach{$(\Event_2,\Event_3)\in E^*$}{
$\mathcal{C}\gets \mathcal{C} \cup \{ (x_{\Event_1, \Event_2} \Rightarrow x_{\Event_1, \Event_3})\}$
}
\ForEach{$(\Event_3, \Event_1)\in E^*$}{
$\mathcal{C}\gets \mathcal{C} \cup \{ (x_{\Event_1, \Event_2} \Rightarrow x_{\Event_3, \Event_2})\}$
}
}
\tcp{3.~Annotation clauses}
\ForEach{$(\Read, \Write)\in \Annotationpos$ and $\Write'\in V\cap \SysWrites$ s.t. $\Confl(\Read, \Write')$}{\label{line:annotation}
$\mathcal{C}\gets \mathcal{C} \cup \{ (x_{\Write', \Read} \Rightarrow x_{\Write', \Write})\}$
}
\tcp{4.~Fact clauses}
\ForEach{$(\Event_1, \Event_2)\in E^*$ with $\Event_1\neq \Event_2$}{\label{line:fact}
$\mathcal{C}\gets \mathcal{C} \cup \{(x_{\Event_1, \Event_2})\}$
}
Compute a satisfying assignment $f:V_{\mathcal{C}}\to \{\False, \True\}^{|V_{\mathcal{C}}|}$ of 
the 2SAT over $\mathcal{C}$, or \Return $\bot$ if $\mathcal{C}$ is unsatisfiable\label{line:2SAT}\\
$E'\gets E \cup \{(\Event_1, \Event_2):~f(x_{\Event_1, \Event_2})=\True\}$\\
Let $G'=(V,E')$\label{line:topological_order}\\
\Return a trace $\Trace$ by topologically sorting the vertices of $G'$\label{line:return_trace}
\end{algorithm}

The formal description of the second phase is given in Algorithm~\ref{algo:realize}.
The following theorem summarizes the results of this section.

\begin{restatable}{theorem}{annotation}\label{them:annotation}
Consider any architecture $\System=(\Process)_i$ and let $\Annotationpos$ be any well-formed positive annotation over a basis $(\SeqTrace)_i$.
Deciding whether $\Annotationpos$ is realizable is \NPC.
If $\System$ is acyclic, the problem can be solved in $O(n^3)$ time, where $n=\sum_i|\SeqTrace_i|$.
\end{restatable}

\section{Data-centric Dynamic Partial Order Reduction}\label{sec:enumerative}

In this section we develop our data-centric DPOR algorithm called 
$\EnumExplore$ and prove its correctness and compactness,
namely that the algorithm explores each observation equivalence class of $\TraceSpace$ \emph{once}.
We start with the notion of causal past cones, which will help in proving the properties of our algorithm.

\subsection{Causal Cones}\label{sec:cones}

Intuitively, the causal past cone of an event $\Event$ appearing in a trace $\Trace$
is the set of events that precede $\Event$ in $\Trace$ and may be responsible for enabling $\Event$ in $\Trace$.

\noindent{\bf Causal cones.}
Given a trace $\Trace$ and some event $\Event\in \Events{\Trace}$, the \emph{causal past cone} $\Past_{\Trace}(\Event)$ of $\Event$ in $\Trace$ is the smallest set that contains the following events:
\begin{compactenum}
\item if there is an event $\Event'\in\Events{\Trace}$ with $\PS(\Event',\Event)$, then $\Event'\in \Past_{\Trace}(\Event)$,
\item if $\Event_1\in \Past_{\Trace}(\Event)$, for every event $\Event_2\in \Events{\Trace}$ such that
$\PS(\Event_2, \Event_1)$, we have that $\Event_2\in \Past_{\Trace}(\Event)$, and
\item if there exists a read event $\Read\in \Past_{\Trace}(\Event)\cap \SysReads$,
we have that $\Obs_{\Trace}(\Read)\in \Past_{\Trace}(\Event)$.
\end{compactenum}
In words, the causal past cone of $\Event$ in $\Trace$ is the set of events $\Event'$ that precede $\Event$ in $\Trace$
and may causally affect the enabling of $\Event$ in $\Trace$.
Note that for every event $\Event'\in \Past_{\Trace}(\Event)$ we have that $\HB{\Event'}{\Trace}{\Event}$,
i.e., every event in the causal past cone of $\Event$ also happens before $\Event$ in $\Trace$.
However, the inverse is not true in general, as e.g. for some read $\Read$ we have
$\HB{\Obs_{\Trace}(\Read)}{\Trace}{\Read}$ but possibly $\Obs_{\Trace}(\Read) \not \in \Past_{\Trace}(\Read)$.

\begin{remark}\label{rem:causal_past}
If $\Event'\in \Past_{\Trace}(\Event)$, then $\HB{\Event'}{\Trace}{\Event}$ and $\Past_{\Trace}(\Event')\subseteq \Past_{\Trace}(\Event)$.
\end{remark}

\begin{remark}\label{rem:casual_projection}
For every trace $\Trace$ and event $\Event\in\Events{\Trace}$ we have that $\Trace\Project (\Past_{\Trace}(\Event)\cup{\Event})$ is a valid trace.
\end{remark}

The following lemma states the main property of causal past cones used throughout the paper.
Intuitively, if the causal past of an event $\Event$ in some trace $\Trace_1$ also appears in another trace $\Trace_2$,
and the read events in the causal past observe the same write events in both traces, then $\Event$ is inevitable in $\Trace_2$,
i.e., every maximal extension of $\Trace_2$ will contain $\Event$.

\begin{restatable}{lemma}{inevitable}\label{lem:obs_to_enable_max}
Consider two traces $\Trace_1, \Trace_2$ and an event $\Event\in \Events{\Trace_1}$ such that 
for every read $\Read\in \Past_{\Trace_1}(\Event)$ we have $\Read\in \Events{\Trace_2}$ and $\Obs_{\Trace_1}(\Read)=\Obs_{\Trace_2}(\Read)$.
Then $\Event$ is inevitable in $\Trace_2$.
\end{restatable}

\subsection{Data-centric Dynamic Partial Order Reduction}\label{subsec:dcdpor}

\noindent{\bf Algorithm $\EnumExplore$.}
We now present our data-centric DPOR algorithm.
The algorithm receives as input a maximal trace $\Trace$ and annotation pair $\Annotation=(\Annotationpos, \Annotationneg)$,
where $\Trace$ is compatible with $\Annotationpos$.
The algorithm scans $\Trace$ to detect conflicting read-write pairs of events that are not annotated, i.e, a read event $\Read\in \Reads{\Trace}$ and a write event $\Write\in\Writes{\Trace}$ such that $\Read\not\in\Domain(\Annotationpos)$ and $\ConflRW(\Read, \Write)$.
If $\Write\not \in \Annotationneg(\Read)$, then $\EnumExplore$ will try to \emph{mutate} $\Read$ to $\Write$, i.e., the algorithm will push $(\Read, \Write)$ in the positive annotation $\Annotationpos$ and call $\Realize$ to obtain a trace that realizes the new positive annotation.
If the recursive call succeeds, then the algorithm will push $\Write$ to the negative annotation of $\Read$,
i.e., will insert $\Write$ to $\Annotationneg(\Read)$.
This will prevent recursive calls from pushing $(\Read, \Write)$ into their positive annotation.
Algorithm~\ref{algo:enumerative} provides a formal description of $\EnumExplore$.
Initially $\EnumExplore$ is executed on input $(\Trace,\Annotation)$ where $\Trace$ is some arbitrary maximal trace,
and $\Annotation=(\emptyset,\emptyset)$ is a pair of empty annotations.

\begin{algorithm}
\small
\DontPrintSemicolon
\caption{$\EnumExplore(\Trace, \Annotation)$}\label{algo:enumerative}
\KwIn{A maximal trace $\Trace$, an annotation pair $\Annotation=(\Annotationpos, \Annotationneg)$}
\BlankLine
\tcp{Iterate over reads not yet mutated}
\ForEach{$\Read\in \Events{\Trace}\setminus \Domain(\Annotationpos)$ in increasing index $\Index_{\Trace}(\Read)$}{\label{line:foreachread}
\tcp{Find conflicting writes allowed by $\Annotationneg$}
\ForEach{$\Write\in \Events{\Trace}$ s.t. $\Confl(\Read, \Write)$ and $\Write\not \in \Annotationneg(\Read)$}{\label{line:foreachwrite}
$\Annotationpos_{\Read, \Write}\gets \Annotationpos\cup \{(\Read, \Write)\}$ \label{line:strengthenpos} \\
\tcp{Attempt mutation and update $\Annotationneg$}
Let $\Trace'\gets \Realize(\Annotationpos_{\Read, \Write})$\label{line:get_trace}\\
\uIf{$\Trace'\neq \bot $}{
$\Trace''\gets$ a maximal extension of $\Trace'$\\
$\Annotationneg(\Read)\gets \Annotationneg(\Read)\cup \{\Write\}$\label{line:strengthenneg}\\
$\Annotation_{\Read, \Write}\gets (\Annotationpos_{\Read, \Write}, \Annotationneg)$\\
Call $\EnumExplore(\Trace'', \Annotation_{\Read,\Write})$ 
}
}
}
\end{algorithm}

We say that $\EnumExplore$ \emph{explores} a class of $\TraceSpace/\sim_{\Obs}$ when it is called on some annotation input $\Annotation=(\Annotationpos, \Annotationneg)$, where $\Annotationpos$ is realized by some (and hence, every) trace in that class.
The representative trace is then the trace $t'$ returned by $\Realize$.
The following two lemmas show the optimality of $\EnumExplore$, namely that the algorithm explores every such class at most once (\emph{compactness}) and at least once (\emph{completeness}).
They both rely on the use of annotations, and the correctness of the procedure $\Realize$ (Theorem~\ref{them:annotation}).
We first state the compactness property, which follows by the use of negative annotations.

\begin{restatable}[Compactness]{lemma}{compactness}
\label{lem:compactness}
Consider any two executions of $\EnumExplore$ on inputs $(\Trace_1, \Annotation_1)$ and $(\Trace_2,\Annotation_2)$.
Then $\Annotationpos_1\neq \Annotationpos_2$.
\end{restatable}

We now turn our attention to completeness, namely that every realizable observation function is realized by a trace explored by $\EnumExplore$.
The proof shows inductively that if $\Trace$ is a trace that realizes an observation function $\Obs$,
then $\EnumExplore$ will explore a trace $\Trace_i$ that agrees with $\Trace$ on the first few read events.
Then, Lemma~\ref{lem:obs_to_enable_max} guarantees that the first read event $\Read$ on which the two traces disagree appears in $\Trace_i$, 
and so does the write event $\Write$ that $\Read$ observes in $\Obs$.
Hence $\EnumExplore$ either will mutate $\Read\to \Write$ (if $\Write\not\in \Annotationneg(\Read)$),
or it has already done so in some earlier steps of the recursion (if $\Write \in \Annotationneg(\Read)$).

\begin{restatable}[Completeness]{lemma}{completeness}
\label{lem:completeness}
For every realizable observation function $\Obs$, $\EnumExplore$ generates a trace $\Trace$ that realizes $\Obs$.
\end{restatable}

We thus arrive to the following theorem.

\begin{restatable}{theorem}{enumexplore}
\label{them:enum_explore}
Consider a concurrent acyclic architecture $\System$ of processes on an acyclic state space,
and $n=\max_{\Trace\in\TraceSpace}|\Trace|$ the maximum length of a trace of $\System$.
The algorithm $\EnumExplore$ explores each class of $\TraceSpace/\sim_{\Obs}$ exactly once,
and requires $O\left(|\TraceSpace/\sim_{\Obs}|\cdot n^5 \right)$ time.
\end{restatable}

We note that our main goal is to explore the exponentially large $\TraceSpace/\sim_{\Obs}$ by spending polynomial time in each class.
The $n^5$ factor in the bound comes from a crude complexity analysis.

\section{Beyond Acyclic Architectures}\label{sec:cyclic_architectures}

In the current section we turn our attention to cyclic architectures.
Recall that according to Theorem~\ref{them:annotation}, procedure $\Realize$ is guaranteed to find a trace that realizes a positive annotation $\Annotationpos$, provided that the underlying architecture is acyclic.
Here we show that the trace space of cyclic architectures can be partitioned wrt an equivalence that is finer than the observation equivalence, but remains (possibly exponentially) coarser than the Mazurkiewicz equivalence.
The current section makes a formal treatment of cyclic architectures with the aim to prove that exponentially coarser equivalences can be used to guide the search.
We refer to our implementation in Section~\ref{sec:cyclic_architectures} for a description of how cyclic architectures are handled in practice.

\noindent{\bf Intuition.}
Recall that the architecture of the concurrent system is an acyclic graph, where nodes represent the processes of the system, 
and two nodes are connected by an edge if the respective processes communicate over a common shared variable.
In high level, our approach for handling cyclic architectures consists of the following steps.
\begin{compactenum}
\item We choose a set of edges $X$ such that removing all edges in $X$ makes the architecture acyclic.
Such a choice can be made arbitrarily.
\item For every variable that is used by at least two processes which have an edge in $X$, we introduce a fresh lock in the system.
\item We transform the concurrent program so that every write event to every such variable is protected by its respective lock.
\end{compactenum}
Intuitively, the new locks have the effect that when calling the procedure $\Realize$ for realizing a positive annotation $\Annotationpos$,
all write events protected by these locks are totally ordered by $\Annotationpos$ (via the lock-acquire and lock-release events).
Hence, $\Realize$ needs only to resolve orderings between read/write events to variables that (by the choice of $X$) create no cycles in the communication graph.

\noindent{\bf Architecture acyclic reduction.}
Consider a cyclic architecture $\System$, and the corresponding communication graph $G_{\System}=(V_{\System}, E_{\System}, \Weight_{\System})$.
We call a set of edges $X\subseteq E_{\System}$ an \emph{all-but-two cycle set} of $G_{\System}$ if every cycle of $G_{\System}$ contains at most two edges outside of $X$.
Given an all-but-two cycle set, $X\subseteq  E_{\System}$
we construct a second architecture $\System^X$, called the \emph{acyclic reduction} of $\System$ over $X$, by means of the following process.
\begin{compactenum}
\item Let $Y=\bigcup_{(\Process_i,\Process_j)\in X}\Weight_{\System}(\Process_i,\Process_j)$ be the set of variables that appear in edges of the set $X$.
We introduce a set of new locks $\ObservedLocks$ in $\System^X$ such that we have exactly one new lock $l_g\in \ObservedLocks$ for each variable $g\in Y$.
\item For every process $\Process_i$, every write event $\Write\in \SysWrites_i$ with $\Location(\Write)\in Y$ is surrounded by an acquire/release pair on the new lock variable $l_{\Location(\Write)}$.
\end{compactenum}

\noindent{\bf Observation equivalence refined by an edge set.}
Consider a cyclic architecture $\System$ and $X$ an edge set of the underlying communication graph $G_{\System}$.
We define a new equivalence on the trace space $\TraceSpace$ as follows.
Two traces $\Trace_1,\Trace_2\in \TraceSpace$ are \emph{observationally equivalent refined by $X$}, denoted by $\sim_{\Obs}^X$, if the following hold:
\begin{compactenum}
\item $\Trace_1\sim_{\Obs} \Trace_2$, and
\item for every edge $(\Process_i,\Process_j) \in X$, for every pair of distinct write events $\Write_1,\Write_2\in \Writes{\Trace_1}\cap (\SysWrites_i \cup \SysWrites_j)$
with $\Location(\Write_1)=\Location(\Write_2)=g$ and $g\in \Weight_{\System}(\Process_i, \Process_j)$,
we have that $\Index_{\Trace_1}(\Write_1)< \Index_{\Trace_1}(\Write_2)$ iff $\Index_{\Trace_2}(\Write_1)< \Index_{\Trace_2}(\Write_2)$
\end{compactenum}

Clearly, $\sim_{\Obs}^X$ refines the observation equivalence $\sim_{\Obs}$.
The following lemma captures that the Mazurkiewicz equivalence refines the observation equivalence refined by an edge set $X$.

\begin{restatable}{lemma}{obsrefined}
\label{lem:obs_refined}
For any two traces $\Trace_1,\Trace_2\in \TraceSpace$, if $\Trace_1\sim_{M} \Trace_2$ then $\Trace_1\sim_{\Obs}^X \Trace_2$.
\end{restatable}

\noindent{\bf Exponential succinctness of $\sim_{\Obs}^X$ in cyclic architectures.}
Here we present a very simple cyclic architecture where the observation equivalence $\sim_{\Obs}^X$ refined by an all-but-two cycle set $X$ is exponentially more succinct than the Mazurkiewicz equivalence $\sim_{M}$.
\begin{figure}[!h]
\centering
\small
\begin{subfigure}[t]{0.10\textwidth}
\begin{align*}
\text{Process}&~\Process_1:\\
\hline\\[-1em]
1.&~\mathsf{write}~x\\
2.&~\mathsf{read}~x\\
\end{align*}
\end{subfigure}
\quad
\begin{subfigure}[t]{0.10\textwidth}
\begin{align*}
\text{Process}&~\Process_2:\\
\hline\\[-1em]
1.&~\mathsf{write}~x\\
2.&~\mathsf{write}~y\\
& \dots\\
n+2.&~\mathsf{write}~y\\
n+3.&~\mathsf{read}~y\\
n+4.&~\mathsf{read}~x
\end{align*}
\end{subfigure}
\quad
\begin{subfigure}[t]{0.10\textwidth}
\begin{align*}
\text{Process}&~\Process_3:\\
\hline\\[-1em]
1.&~\mathsf{write}~x\\
2.&~\mathsf{write}~y\\
& \dots\\
n+2.&~\mathsf{write}~y\\
n+3.&~\mathsf{read}~y\\
n+4.&~\mathsf{read}~x
\end{align*}
\end{subfigure}
\caption{A cyclic architecture of three processes.}
\label{fig:m_comparison_cyclic}
\end{figure}

Consider the architecture $\System$ in Figure~\ref{fig:m_comparison_cyclic}, which consists of three processes and two single global variables $x$ and $y$.
We choose an edge set as $X=\{(\Process_1, \Process_2)\}$, and $X$ is an all-but-two cycle set of $G_{\System}$.
We argue that $\sim_{\Obs}^X$ is exponentially more succinct than $\sim_{M}$ by showing exponentially many traces which are pairwise equivalence under $\sim_{\Obs}^X$ but not under $\sim_{M}$.
Indeed, consider the set $T$ which consists of all traces such that the following hold
\begin{compactenum}
\item All traces start with $\Process_1$ executing to completion, then $\Process_2$ executing its first statement, and $\Process_3$ executing its first statement.
\item All traces end with the last three events of $\Process_2$ followed by the last two events of $\Process_3$.
\end{compactenum}
Note that $|T|=\binom{2 \cdot n}{n}$ as there are $(2\cdot n)!$ ways to order the $2\cdot n$ $\mathsf{write}~y$ events of the two processes,
but $n!\cdot n!$ orderings are invalid as they violate the program structure.
All traces in $T$ have the same observation function, yet they are inequivalent under $\sim_{M}$ since every pair of them orders two $\mathsf{write}~y$ events differently. 
Finally, $\TraceSpace/\sim_{\Obs}^X$ is only exponentially large, and since 
\[
|(\TraceSpace/\sim_{M}) \setminus (\TraceSpace/\sim_{\Obs}^x) | \geq |T|-1
\]
we have that $\sim_{\Obs}^X$ is exponentially more succinct than $\sim_M$.

\noindent{\bf Data-centric DPOR on a cyclic architecture.}
We are now ready to outline the steps of the data-centric DPOR algorithm on a cyclic architecture $\System$,
called $\EnumExploreCyclic$.
First, we determine an all-but-two cycle set $X$ of the underlying communication graph $G_{\System}=(V_{\System}, E_{\System}, \Weight_{\System})$,
and construct the acyclic reduction $\System^X$ of $\System$ over $X$.
The set $X$ can be chosen arbitrarily, e.g. by letting $|X|=|E_{\System}|-2$ (i.e., adding in $X$ all the edges of $G_{\System}$ except for two).
Then, we execute $\EnumExplore$ on $\System^X$, with the following two modifications on the procedure $\Realize$.
\begin{compactenum}
\item\label{item:mod1} Consider the graph $G=(V,E)$ constructed in Line~\ref{line:graph_G} of $\Realize$ (Algorithm~\ref{algo:realize}).
For every pair of write events $\Write, \Write'$ protected by some of the new locks $\ell\in \ObservedLocks$,
for every read event $\Read$ such that $\Annotationpos(\Read)=\Write$, if $(\Write, \Write')\in E$ then we add an edge $(\Read, \Write)$ in $E$,
and if $(\Write',\Read)\in E$, then we add an edge $(\Write', \Write)$ in $E$.
\item\label{item:mod2} If at the end of Item~\ref{item:mod1} $G$ has a cycle, $\Realize$ returns $\bot$.
\item\label{item:mod3} In Line~\ref{line:variables} we use the edge set $E_{\System^X}\setminus X$.
Hence for every variable $x_{\Event_1, \Event_2}$ used in the 2SAT reduction, we have either $\Proc(\Event_1)=\Proc(\Event_2)$ or $(\Proc(\Event_1), \Proc(\Event_2))\in E_{\System^X}\setminus X$.
\end{compactenum}

We arrive at the following theorem.

\begin{restatable}{theorem}{enumexplorecyclic}
\label{them:enum_explore_cyclic}
Consider a concurrent architecture $\System$ of processes on an acyclic state space,
and $n=\max_{\Trace\in\TraceSpace}|\Trace|$ the maximum length of a trace of $\System$.
Let $X$ be an all-but-two cycle set of the communication graph $G_{\System}$.
The algorithm $\EnumExploreCyclic$ explores each class of $\TraceSpace/\sim_{\Obs}^{X}$ exactly once,
and requires $O\left(|\TraceSpace/\sim_{\Obs}^{X}|\cdot n^5 \right)$ time.
\end{restatable}

Theorem~\ref{them:enum_explore_cyclic} establishes that for cyclic architectures, $\EnumExploreCyclic$ explores a partitioning of the trace space that is coarser than the Mazurkiewicz partitioning, and spends only polynomial time per class. 
We refer to our implementation in Section~\ref{sec:cyclic_architectures} for a description of how cyclic architectures are handled in practice.
\section{Experiments}\label{sec:experiments}

Here we report on the implementation and experimental evaluation of our data-centric DPOR algorithm. 

\subsection{Implementation Details}\label{subsec:implementation}

\noindent{\bf Implementation.}
We have implemented our data-centric DPOR in C++,
by extending the tool Nidhugg\footnote{\url{https://github.com/nidhugg/nidhugg}}.
Nidhugg is a powerful tool that utilizes the LLVM compiler infrastructure, 
and hence our treatment of programs is in the level of LLVM's intermediate representation (IR).
Concurrent architectures are supported via POSIX threads.

\noindent{\bf Handling static arrays.}
The challenge in handling arrays (and other data structures) lies in the difficulty of determining whether two global events access the same location of the array (and thus are in conflict) or not.
Indeed, this is not evident from the CFG of each process, but depends on the values of indexing variables (e.g. the value of local variable $i$ in an access to $\mathsf{table}[i]$).
DPOR methods offer increased precision, as during the exploration of the trace space, backtracking points are computed dynamically, given a trace, where the value of indexing variables is known.
In our case, the value of indexing variables is also needed when procedure $\Realize$ is invoked to construct a trace which realizes a positive annotation $\Annotationpos$. 
Observe that the values of all such variables are determined by the value function $\Value_{\Annotationpos}$, and thus in every sequential trace $\SeqTrace_i$ of the basis $(\SeqTrace_i)_i$ of $\Annotationpos$ these values are also known.
Hence, arrays are handled naturally by the dynamic flavor of the exploration.

\noindent{\bf Handling cyclic architectures.}
In order to effectively handle cyclic architectures, we followed the following process.
Wlog, we considered that the input architecture always has the most difficult topology, namely it is a clique.
First, the cyclic architecture is converted to a star architecture, by choosing some distinguished process $\Process_1$ as the root of the star,
and the remaining processes $\Process_2,\dots \Process_k$ are the leaves.
Recall that a a positive annotation yields a sequential trace for each process.
We use the Mazurkiewicz equivalence to generate all possible Mazurkiewicz-based interleavings between traces of the leaf processes,
and our observation equivalence to generate all possible observation-based interleavings between the root and every leaf process.
Hence the observation equivalence is wrt the star sub-architecture, which is acyclic, and thus our techniques from Theorem~\ref{them:enum_explore} are applicable.
We note that since the Mazurkiewicz interleavings always have to be generated among sequential traces (i.e., straight-line programs),
we are generating them optimally (i.e., obtaining exactly one trace per Mazurkiewicz class) easily, using vector clocks~\cite{Mattern89}.
See Figure~\ref{fig:clyclic_star} for an illustration.

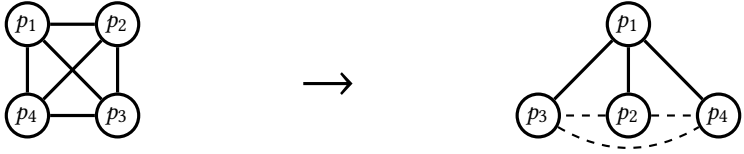
\begin{figure}
\newcommand{\distone}{1cm}
\small
\centering
\begin{tikzpicture}[->,>=stealth',shorten >=0pt,auto,node distance=\distone,
                    semithick,scale=1 ]
      
\tikzstyle{every state}=[fill=white,draw=black,text=black,font=\small, thick, inner sep=-1cm, minimum size=0.15cm]
\tikzstyle{invis}=[fill=white,draw=white,text=white,font=\small , inner sep=-0.05cm]

\newcommand{\xdisposition}{0}
\newcommand{\ydisposition}{0}
\newcommand{\ystep}{1.2}
\newcommand{\xstep}{1.2}

\node[circle, draw=black, very thick, minimum size=0.05mm, inner sep=2] (p1) at (\xdisposition + 0 * \xstep, \ydisposition + 0 * \ystep) {$\Process_1$};
\node[circle, draw=black, very thick, minimum size=0.05mm, inner sep=2]  (p2) at (\xdisposition + 1 * \xstep, \ydisposition + 0 * \ystep) {$\Process_2$};
\node[circle, draw=black, very thick, minimum size=0.05mm, inner sep=2] (p3) at (\xdisposition + 1 * \xstep, \ydisposition + -1 * \ystep) {$\Process_3$};
\node[circle, draw=black, very thick, minimum size=0.05mm, inner sep=2] (p4) at (\xdisposition + 0 * \xstep, \ydisposition + -1 * \ystep) {$\Process_4$};

\draw[-, very thick] (p1) to (p2);
\draw[-, very thick] (p1) to (p3);
\draw[-, very thick] (p1) to (p4);
\draw[-, very thick] (p2) to (p3);
\draw[-, very thick] (p2) to (p4);
\draw[-, very thick] (p3) to (p4);

\node[] at (4,-0.8) {\Huge$\rightarrow$};

\renewcommand{\xdisposition}{8}

\node[circle, draw=black, very thick, minimum size=0.05mm, inner sep=2] (q1) at (\xdisposition + 0 * \xstep, \ydisposition + 0 * \ystep) {$\Process_1$};
\node[circle, draw=black, very thick, minimum size=0.05mm, inner sep=2]  (q2) at (\xdisposition + 0 * \xstep, \ydisposition - 1 * \ystep) {$\Process_2$};
\node[circle, draw=black, very thick, minimum size=0.05mm, inner sep=2] (q3) at (\xdisposition -1 * \xstep, \ydisposition + -1 * \ystep) {$\Process_3$};
\node[circle, draw=black, very thick, minimum size=0.05mm, inner sep=2] (q4) at (\xdisposition + 1 * \xstep, \ydisposition + -1 * \ystep) {$\Process_4$};

\draw[-, very thick] (q1) to (q2);
\draw[-, very thick] (q1) to (q3);
\draw[-, very thick] (q1) to (q4);

\draw[-,  thick, dashed] (q2) to (q3);
\draw[-,  thick, dashed] (q2) to (q4);
\draw[-,  thick, bend right=30, dashed] (q3) to (q4);

\end{tikzpicture}
\caption{Converting a cyclic architecture to a star architecture which is acyclic.
On the star, solid edges correspond to observation equivalence interleavings, and dashed edges correspond to Mazurkiewicz equivalence interleavings.}
\label{fig:clyclic_star}
\end{figure}

\noindent{\bf Optimizations.}
Since our focus is on demonstrating a new, data-centric principle of DPOR, we focused on a basic implementation and avoided engineering optimizations.
We outline two straightforward algorithmic optimizations which were simple and useful.
\begin{compactenum}
\item (\emph{Burst mutations}). 
Instead of performing one mutation at a time, the algorithm performs a sequence of several mutations at once.
In particular, given a trace $\Trace$, any time we want to add a pair $(\Read, \Write)$ to the positive annotation,
we also add $(\Read', \Obs_{\Trace}(\Read'))$, where $\Read'\in \Past_{\Trace}(\Read)\cup \Past_{\Trace}(\Write)$
ranges over all read events in the causal past of $\Read$ and $\Write$ in $\Trace$.
This makes the recursion tree shallower, as now we do not need to apply any mutation $(\Read, \Write)$, where $\Write=\Obs_{\Trace}(\Read)$, individually.
\item (\emph{Cycle detection}). As a preprocessing step, before executing procedure $\Realize$ on some positive annotation $\Annotationpos$ input,
we test whether the graph $G$ (in Line~\ref{line:graph_G}) already contains a cycle.
The existence of a cycle is a proof that $\Annotationpos$ is not realizable, and requires linear instead of cubic time, as the graph is sparse.
\end{compactenum}

\subsection{Experimental Results}\label{subsec:experiments}

We now turn our attention to the experimental results.
Our comparison is with the Source-DPOR algorithm from~\cite{Abdulla14} and the tool Nidhugg that implements it~\cite{Abdulla2015}.
To our knowledge, Source-DPOR is the latest and state-of-the-art DPOR which has implemented for C programs.

\noindent{\bf Experimental setup.}
In our experiments, we have compared our data-centric DPOR, with the Mazurkiewicz-based Source-DPOR  introduced recently in~\cite{Abdulla14} as an important improvement over the traditional DPOR~\cite{Flanagan05}.
Our benchmark set consists of synthetic benchmarks, as well as benchmarks obtained from the TACAS Competition on Software Verification (SV-COMP).
Most of the benchmarks have tunable size, by specifying a loop-unroll bound, or the number of threads running in parallel.
In all cases, we compared th running time and number of traces explored by $\EnumExplore$ and Source-DPOR.
We have set a timeout of 1 hour.
All benchmarks were executed on an Ubuntu-based virtual machine, given 4GB of memory and one 2GHz CPU.

\noindent{\bf Two synthetic benchmarks.}
First we analyze the two synthetic benchmarks \emph{lastzero} and \emph{opt\_lock} found in Table~\ref{tab:experiments_lastzero} and Table~\ref{tab:experiments_optlock}, respectively.
The benchmark \emph{lastzero} was introduced in~\cite{Abdulla14} to demonstrate the superiority of Source-DPOR over the traditional DPOR from~\cite{Flanagan05}.
It consists of $n$ threads writing to an array, and $1$ thread reading from it.
We observe that our $\EnumExplore$ explores exponentially fewer traces than Source-DPOR.
In fact, the number of traces explored by our data-centric approach scales polynomially, whereas the number explored by the Mazurkiewicz-based approach grows exponentially with the number of threads.
Consequently, our $\EnumExplore$ runs much faster, and manages to scale on larger input sizes.
We note that the number of traces explored from Source-DPOR differs from the number reported in~\cite{Abdulla14}.
This is natural as the implementation of~\cite{Abdulla14} handles programs written in Erlang, a functional language with concurrency mechanisms much different from C
\footnote{We later noticed that our implementation of the benchmark differs slightly from that given in~\cite{Abdulla14}, which might also account for the difference in the reported traces}.

The benchmark \emph{\emph{opt\_lock}} mimics an optimistic locking scheme of 2 threads.
Each thread tries to update some variable, and afterwards checks if it was interrupted.
If not, it terminates, otherwise it tries again, up to a total $n$ number of attempts.
Again, we see that the number of explored traces by $\EnumExplore$ grows polynomially, whereas 
the number explored by Source-DPOR grows exponentially. 
Hence, our algorithm manages to handle much larger input sizes than the Mazurkiewicz-based Source-DPOR.
Recall that, as Theorem~\ref{them:enum_explore} states, this exponential reduction in the explored traces comes with polynomial-time guarantees per trace.

\begin{minipage}[t]{0.48\textwidth}
\begin{algorithm}[H]
\footnotesize
\SetKwInOut{Globals}{Globals}
\SetKwInOut{Locals}{Locals}
\setcounter{algocf}{-1}
\renewcommand{\thealgocf}{}
\SetAlgorithmName{\emph{lastzero}($n$)}{algorithm}
\SetAlgoNoLine
\DontPrintSemicolon
\caption{$n+1$ processes}
\Globals{
$\mathit{int}~\mathsf{array}[n+1]$\\
}
\BlankLine
\tcp{--------- Process $j=0$ ---------}
\Locals{$\mathit{int}~i$}
$i\gets n$\\
\While{$\mathsf{array}[i]\neq 0$}{
$i\gets i-1$
}
\BlankLine
\BlankLine
\BlankLine
\tcp{--------- Process $1< j\leq n$ ---------}
$\mathsf{array}[j]\gets \mathsf{array}[j-1] + 1$
\BlankLine
\BlankLine
\BlankLine
\end{algorithm}
\end{minipage}
\qquad
\begin{minipage}[t]{0.48\textwidth}
\begin{algorithm}[H]
\footnotesize
\SetKwInOut{Globals}{Globals}
\SetKwInOut{Locals}{Locals}
\setcounter{algocf}{-1}
\renewcommand{\thealgocf}{}
\SetAlgorithmName{\emph{opt\_lock}($n$)}{algorithm}
\SetAlgoNoLine
\DontPrintSemicolon
\caption{$2$ processes and $n$ attempts}
\Globals{
$\mathit{int}~\mathsf{last\_id}, x$\\
}
\BlankLine
\tcp{--------- Process $0< j<2$ ---------}
\Locals{$\mathit{int}~i$}
$i\gets 0$\\
\While{$i< n$}{
$i\gets i+1$\\
$\mathsf{last\_id}\gets j$\\
$x\gets \mathsf{get\_message}(j)$\\
\uIf{$\mathsf{last\_id}=j$}{
\Return $x$
}
}
\Return -1
\BlankLine
\BlankLine
\BlankLine
\end{algorithm}
\end{minipage}

\begin{table}
\caption{Experimental results on two synthetic benchmarks.}
\label{tab:synthetic_benchmarks}
\begin{subfigure}[t]{.49\textwidth}
\footnotesize
\setlength\tabcolsep{1pt}
\caption{Experiments on lastzero($n$), for $n+1$ threads.
'-' indicates a timeout after 1 hour.}
\label{tab:experiments_lastzero}
\begin{tabular}{|c||c c|c c|}
\hline
\textbf{Benchmark} & \multicolumn{2}{c|}{\textbf{Traces}} & \multicolumn{2}{c|}{\textbf{Time (s)}}\\
\hline
& DC-DPOR  & S-DPOR & DC-DPOR  & S-DPOR \\
\hline
\hline
lastzero(4) & 38  & 2,118 & 0.21 & 0.84\\
\hline
lastzero(5) & 113 & 53,172 & 0.34 & 19.29\\ 
\hline
lastzero(6) &  316 & 1,765,876 & 0.63 & 856 \\
\hline
lastzero(7) & 937 & \footnotesize{-} & 1.8 & \footnotesize{-}\\
\hline
lastzero(8) & 3,151 & \footnotesize{-} & 9.32 &  \footnotesize{-}\\
\hline
lastzero(9) & 12,190 & \footnotesize{-} & 47.97 & \footnotesize{-}\\
\hline
lastzero(10) & 52,841  & \footnotesize{-} & 383.12 &\footnotesize{-}\\
\hline
\end{tabular}
\end{subfigure}
\quad
\begin{subfigure}[t]{.47\textwidth}
\footnotesize
\setlength\tabcolsep{1pt}
\caption{Experiments on opt\_lock($n$), where $n$ is the number of attempts to optimistically lock.
'-' indicates a timeout after 1 hour.}
\label{tab:experiments_optlock}
\begin{tabular}{|c||c c|c c|}
\hline
\textbf{Benchmark} & \multicolumn{2}{c|}{\textbf{Traces}} & \multicolumn{2}{c|}{\textbf{Time (s)}}\\
\hline
& DC-DPOR  & S-DPOR & DC-DPOR  & S-DPOR \\
\hline
\hline
opt\_lock(12) & 141 & 785,674 & 0.35 &  252.64 \\
\hline
opt\_lock(13) & 153 & 2,056,918 & 0.36 &  703.90 \\
\hline
opt\_lock(14) & 165 & 5,385,078  & 0.43    & 1,880.12 \\
\hline
opt\_lock(15) & 177 & \footnotesize{-}  & 0.46   & \footnotesize{-} \\
\hline
opt\_lock(50) & 597 & \footnotesize{-} & 5.91 & \footnotesize{-}\\
\hline
opt\_lock(100) & 1,197 & \footnotesize{-} & 43.82 & \footnotesize{-}\\
\hline
opt\_lock(200) & 2,397 & \footnotesize{-} & 450.99 & \footnotesize{-}\\
\hline
\end{tabular}
\end{subfigure}
\end{table}

\noindent{\bf Benchmarks from SV-COMP.}
We now turn our attention to benchmarks from SV-COMP, namely 
\emph{fib\_bench}, \emph{pthread\_demo}, \emph{sigma\_false} and \emph{parker},
which are found in Table~\ref{tab:experiments_fib_bench}, Table~\ref{tab:experiments_pthread}, Table~\ref{tab:experiments_sigma} and Table~\ref{tab:experiments_parker}, respectively.
Similarly to our findings on the synthetic benchmarks, the data-centric $\EnumExplore$ manages to explore fewer traces than the Mazurkiewicz-based Source-DPOR.
In almost all cases, our algorithm run much faster, offering exponential gains in terms of time.
One exception is the benchmark \emph{parker}, where our $\EnumExplore$ is slower.
Although the number of traces explored is less than that of Source-DPOR,
the latter method managed to spend less time in discovering each trace, which led to a smaller overall time.
We note, however, that the improvement of Source-DPOR over $\EnumExplore$ appears to grow only as a small polynomial wrt the input size.
Recall that new traces are discovered by $\EnumExplore$ using the procedure $\Realize$, which can take cubic time in the worst case (Theorem~\ref{them:annotation}).
Hence, we identify optimizations to $\Realize$ as an important challenge that will contribute further to the scalability of our approach.

\begin{table}
\caption{Experimental results on four benchmarks from SV-COMP.}
\label{tab:svcomp_benchmnarks}
\begin{subfigure}[t]{.49\textwidth}
\vspace{0pt}
\begin{flushleft}
\footnotesize
\setlength\tabcolsep{1pt}
\caption{Experiments on fib\_bench($n$), where $n$ is the loop-unroll bound.}
\label{tab:experiments_fib_bench}
\begin{tabular}{|c||c c|c c|}
\hline
\textbf{Benchmark} & \multicolumn{2}{c|}{\textbf{Traces}} & \multicolumn{2}{c|}{\textbf{Time (s)}}\\
\hline
& DC-DPOR  & S-DPOR & DC-DPOR  & S-DPOR \\
\hline
\hline
fib\_bench(4) & 1,233 & 19,605 & 0.93  & 3.03 \\
\hline
fib\_bench(5) & 8,897 & 218,243 & 7.41 & 37.82 \\
\hline
fib\_bench(6) & 70,765  & 2,364,418  & 85.71  & 463.52 \\
\hline
\end{tabular}
\end{flushleft}
\end{subfigure}
\begin{subfigure}[t]{.49\textwidth}
\vspace{0pt}
\begin{flushright}
\footnotesize
\setlength\tabcolsep{1pt}
\caption{Experiments on pthread\_demo($n$), where $n$ is the loop-unroll bound.}
\label{tab:experiments_pthread}
\begin{tabular}{|c||c c|c c|}
\hline
\textbf{Benchmark} & \multicolumn{2}{c|}{\textbf{Traces}} & \multicolumn{2}{c|}{\textbf{Time (s)}}\\
\hline
& DC-DPOR  & S-DPOR & DC-DPOR  & S-DPOR \\
\hline
\hline
pthread\_demo(8) & 256 & 12,870 & 0.37 & 3.17 \\
\hline
pthread\_demo(10) & 1,024 & 184,756 & 1.23 & 49.51 \\
\hline
pthread\_demo(12) & 4,096 & 2,704,156 & 5.30 & 884.99 \\
\hline
\end{tabular}
\end{flushright}
\end{subfigure}
\\
\begin{subfigure}[t]{.49\textwidth}
\begin{flushleft}
\footnotesize
\setlength\tabcolsep{1pt}
\caption{Experiments on sigma\_false($n$), where $n$ is the loop-unroll bound.
'-' indicates a timeout after 1 hour.}
\label{tab:experiments_sigma}
\begin{tabular}{|c||c c|c c|}
\hline
\textbf{Benchmark} & \multicolumn{2}{c|}{\textbf{Traces}} & \multicolumn{2}{c|}{\textbf{Time (s)}}\\
\hline
& DC-DPOR  & S-DPOR & DC-DPOR  & S-DPOR \\
\hline
\hline
sigma\_false(6) & 16 & 10,395 & 0.22 & 2.57 \\
\hline
sigma\_false(7) & 22 & 135,135 &  0.26 & 38.41 \\
\hline
sigma\_false(8) & 29 & 2,027,025 & 0.28 & 658.27 \\
\hline
sigma\_false(9) & 37 & - & 0.38 & - \\
\hline
sigma\_false(10) & 46 & - & 0.44 & -\\
\hline
\end{tabular}
\end{flushleft}
\end{subfigure}
\quad
\begin{subfigure}[t]{.45\textwidth}
\begin{flushright}
\footnotesize
\setlength\tabcolsep{1pt}
\caption{Experiments on parker($n$), where $n$ is the loop-unroll bound.}
\label{tab:experiments_parker}
\begin{tabular}{|c||c c|c c|}
\hline
\textbf{Benchmark} & \multicolumn{2}{c|}{\textbf{Traces}} & \multicolumn{2}{c|}{\textbf{Time (s)}}\\
\hline
& DC-DPOR  & S-DPOR & DC-DPOR  & S-DPOR \\
\hline
\hline
parker(8) & 1,254 & 3,343 & 1.52 & 1.33 \\
\hline
parker(10) & 2,411 & 6,212 & 5.03 & 3.96\\
\hline
parker(12) & 4,132 & 10,361 & 8.09  & 5.62\\
\hline
parker(14) & 6,529 & 16,022  & 11.96  & 6.86\\
\hline
parker(16) & 9,714 & 23,427  & 19.89  & 10.85\\
\hline
\end{tabular}
\end{flushright}
\end{subfigure}
\end{table}

\section{Related Work}\label{sec:related}
The analysis of concurrent programs is a major challenge in program analysis and verification, 
and has been a subject of extensive study~\cite{Petri62,Cadiou73,Lipton75,Clarke86,Lal09,Farzan12,Farzan09}.
The hardness of reproducing bugs by testing, due to scheduling non-determinism,
makes model checking a very relevant approach~\cite{Godefroid05,Musuvathi07,Andrews04,Clarke00,Alglave13},
and in particular stateless model checking to combat the state-space explosion.
To combat the exponential number of interleaving explosion faced by the 
early model checking~\cite{Godefroid97},
several reduction techniques have been proposed such as POR and 
context bounding~\cite{Peled93,Musuvathi07}.
Several POR methods, based on persistent set~\cite{Clarke99,Clarke99,G96,Valmari91}
and sleep set techniques~\cite{Godefroid97},
have been explored.  
DPOR~\cite{Flanagan05} presents on-the-fly construction of persistent sets,
and several variants and improvements have been considered~\cite{Sen06,Lauterburg10,Tasharofi12,Saarikivi12,Sen06b}.
In~\cite{Abdulla14}, source sets and wakeup trees techniques were developed to make DPOR optimal, in the sense that 
the enumerative procedures explores exactly one representative from each Mazurkiewicz class.
Other important works include normal form representation of concurrent 
executions~\cite{Kahlon09} using SAT or SMT-solvers;
or using unfoldings for optimal reduction in number of interleavings~\cite{Kahkonen12,McMillan95,Sousa15}.
Techniques for transition-based POR for message passing programs have 
also been considered~\cite{G96,Godefroid95,Katz92},
and some works extend POR to relaxed memory models~\cite{Wang08,Abdulla2015}.

Another direction of DPOR is SAT/SMT-based, such as Maximal Causality Reduction (MCR)~\cite{HUANG15}, and SATCheck~\cite{Demsky15}. 
Such techniques require an NP oracle to guide each step of the search (i.e., a SAT/SMT solver), and thus suffer scalability issues.
On the other hand, they have the possibility of exploring fewer traces than the traditional DPOR methods.
Among these works, MCR~\cite{HUANG15} is closer to ours, hence we make a more extensive reference to it.

\noindent{\bf Comparison with~\cite{HUANG15}.}
Maximal Causality Reduction is a form of partial-order reduction that is based on coarsening the Mazurkiewicz equivalence.
The main principle is to try and explore traces in which read events observe different values (as opposed to different write events as in our case).
As a result, MCR can potentially create an equivalence that is coarser than the observation equivalence that we introduce here.
However, the MCR approach is very different from ours, in the following three important aspects.
\begin{compactenum}
\item MCR uses an SMT solver to explore each class of the partitioning.
In other words, the MCR relies on an NP-oracle, which has exponential worst-case complexity, even for exploring a single class of the partitioning.
In contrast, our approach spends provably polynomial time per class.
We also note that the experimental results of~\cite[Page~9]{HUANG15} identify the SMT procedure as the bottleneck of the whole approach, and ask for an efficient method for each explored trace. 
\item MCR is not optimal wrt to its partitioning. In fact, many equivalence classes of the partitioning can be explored exponentially many times.
We provide here a minimal example.
Consider the program depicted in Figure~\ref{subfig:mcr_processes}.
We have two processes $\Process_1, \Process_2$, and two global variables $x,y$.
The first processes is $\Process_1=\Write_y^1 \Read_x^1$, and the second process is $\Process_2=\Write_x^2 \Read_y^2$.
Additionally, we let $\Write_x^0$ and $\Write_y^0$ denote the two initialization write events.

\begin{figure}[!h]
\begin{subfigure}[b]{.48\textwidth}
\centering
\begin{subfigure}[t]{0.10\textwidth}
\begin{align*}
\text{Process}&~\Process_1:\\
\hline\\[-1em]
1.&~\mathsf{write}~y\\
2.&~\mathsf{read}~x\\
\end{align*}
\end{subfigure}
\quad
\begin{subfigure}[t]{0.10\textwidth}
\begin{align*}
\text{Process}&~\Process_2:\\
\hline\\[-1em]
1.&~\mathsf{write}~x\\
2.&~\mathsf{read}~y\\
\end{align*}
\end{subfigure}
\caption{}\label{subfig:mcr_processes}
\end{subfigure}
\begin{subfigure}[b]{.48\textwidth}
\centering
\begin{tikzpicture}[thick, >=latex,
pre/.style={<-,shorten >= 1pt, shorten <=1pt, thick},
post/.style={->,shorten >= 1pt, shorten <=1pt,  thick},
und/.style={very thick, draw=gray},
node1/.style={circle, minimum size=4mm, draw=black!100, line width=1pt, inner sep=0},
node2/.style={circle, minimum size=5mm, draw=black!100, fill=white!100, very thick, inner sep=0},
virt/.style={circle,draw=black!50,fill=black!20, opacity=0}]

\newcommand{\xdisposition}{0}
\newcommand{\ydisposition}{0}
\newcommand{\xstep}{1}
\newcommand{\ystep}{0.9}

\node[node1] (t0) at (\xdisposition  + 0 * \xstep, \ydisposition + 0*\ystep ){$a$};
\node[node1] (tx) at (\xdisposition  + -1 * \xstep, \ydisposition + -1*\ystep ){$b$};
\node[node1] (ty) at (\xdisposition  + 1 * \xstep, \ydisposition + -1*\ystep ){$c$};
\node[node1] (txy) at (\xdisposition  + -1 * \xstep, \ydisposition + -2*\ystep ){$d$};
\node[node1] (tyx) at (\xdisposition  + 1 * \xstep, \ydisposition + -2*\ystep ){$e$};

\draw[->, very thick] (t0) to node[left, label={[yshift=-6, xshift=-15]$\Read_x^1\to \Write_x^2$}]{} (tx);
\draw[->, very thick] (t0) to node[right, label={[yshift=-6, xshift=15]$\Read_y^2\to \Write_y^1$}]{} (ty);
\draw[->, very thick] (tx) to node[left]{$\Read_y^2\to \Write_y^1$} (txy);
\draw[->, very thick] (ty) to node[right]{$\Read_x^1\to \Write_x^2$} (tyx);

\end{tikzpicture}
\caption{}\label{subfig:mcr_recursion}
\end{subfigure}
\caption{A simple concurrent system of two processes, and the corresponding MCR exploration.}\label{fig:mcr_comparison}
\end{figure}

We now consider the MCR exploration of the above system.

\begin{compactenum}
\item Initially, the process starts in node $a$ with an empty seed interleaving.
The process will generate two seed interleavings $\{b,c\}$, forcing $\Read_x^1$ to observe the value of $\Write_x^2$, and $\Read_y^2$ to observe the value of $\Write_y^1$, respectively.
\item When in $b$, the process will create the seed interleaving $\{d\}$, forcing $\Read_y^2$ to observe the value of $\Write_y^1$.
\item When in $c$, the process will create the seed interleaving $\{e\}$, forcing $\Read_x^1$ to observe the value of $\Write_x^2$.
\end{compactenum}
Hence, all traces represented in the leaves $\{d, e\}$ of the above tree have the same observation function (and thus all reads observe the same values).
In contract, the optimality of our algorithm guarantees that the exploration will never explore both $d$ and $e$.

The example can easily be generalized to one where the MCR will explore the same class exponentially many times.
The only principle necessary to make the example work is that different branches of the recursion can accumulate the same read-to-write observations in different order, leading to the same read-to-write observations overall.
\end{compactenum}
Hence, compared to our approach, MCR suffers an exponentiation in complexity in two parts: (i)~visiting each class of the partitioning exponentially many times, and (ii)~using an NP-oracle with exponential worst-case behavior for each visit.

\section{Conclusions}\label{sec:conclusions}
We introduce the new observation equivalence on traces that refines the Mazurkiewicz equivalence and can even be exponentially more succinct. 
We develop an optimal, data-centric DPOR algorithm for acyclic architectures based on this new 
equivalence, and also extend a finer version of it to cyclic architectures.
There are several future directions based on the current work.
First, it is interesting to determine whether other, coarser equivalence classes can be  developed for cyclic architectures,
which can be used by some enumerative exploration of the trace space.
Another promising direction is phrasing our observation equivalence on other memory models and developing DPOR algorithms for such models.
Finally, it would be interesting to explore the engineering challenges in applying our approach to real-life examples,
such as using static analysis to obtain the best way for transforming cyclic architectures to acyclic.


\begin{acks}
The research was partly supported by Austrian Science Fund (FWF) Grant No P23499- N23, 
FWF NFN Grant No S11407-N23 (RiSE/SHiNE), ERC Start grant (279307: Graph Games), 
and Czech Science Foundation grant GBP202/12/G061.
\end{acks}

\bibliography{bibliography}




\end{document}